\documentclass{article}
\usepackage[utf8]{inputenc}
\usepackage{amsmath,amssymb,amsfonts,amsthm}
\usepackage{mathrsfs}
\usepackage{indentfirst}
\usepackage[pdftex]{color,graphicx}
\usepackage[pdftex,bookmarks,unicode,colorlinks]{hyperref}
\usepackage{enumerate} 
\usepackage[numbers]{natbib} 
\usepackage{yfonts} 
\usepackage[toc,page]{appendix}

\newcommand{\widesim}[2][3.0]{
  \mathrel{\overset{#2}{\scalebox{#1}[1]{$\sim$}}}
}

\usepackage{bigfoot}
\DeclareNewFootnote{A}[arabic]
\DeclareNewFootnote{B}[fnsymbol]
\usepackage{etoolbox}
\makeatletter
\patchcmd\maketitle{\def\@makefnmark{\rlap{\@textsuperscript{\normalfont\@thefnmark}}}}{}{}{}
\makeatother
\makeatletter
\def\thanksA#1{
  \footnotemarkA\protected@xdef\@thanks{\@thanks%
        \protect\footnotetextA[\the \c@footnoteA]{#1}}%
}
\def\thanksB#1{%
  \footnotemarkB%
  \protected@xdef\@thanks{\@thanks%
        \protect\footnotetextB[\the \c@footnoteB]{#1}}%
}
\makeatother

\theoremstyle{plain}
\newtheorem{theorem}{Theorem}[section]
\newtheorem{proposition}[theorem]{Proposition}
\newtheorem{corollary}[theorem]{Corollary}
\newtheorem{lemma}[theorem]{Lemma}

\theoremstyle{remark}
\newtheorem{remark}[theorem]{Remark}
\newtheorem{example}[theorem]{Example}

\theoremstyle{definition}
\newtheorem{definition}[theorem]{Definition}
\newtheorem{assumption}{Assumption}

\newcommand{\R}{\mathbb R}
\newcommand{\N}{\mathbb N}

\newcommand{\e}{\epsilon}
\newcommand{\F}{\mathcal F}

\newcommand{\ind}{\mathbf 1}

\newcommand{\red}[1]{\textcolor[rgb]{1.00,0.00,0.00}{{#1}}}

\numberwithin{equation}{section} 
\usepackage{subfig}

\begin{document}
\title{\bf\Large 
The Most Probable Transition Paths of Stochastic Dynamical
Systems: A Sufficient and Necessary Characterization}
\author{
\bf\normalsize{
Yuanfei Huang\thanksA{School of Mathematics and Statistics \& Center for Mathematical Sciences \& Hubei Key Laboratory of Engineering Modeling and Scientific Computing, Huazhong University of Science and Technology, Wuhan 430074, P.R. China. 
}
$^{,}$\thanksA{Department of Statistics and Data Science, National University of Singapore, 6 Science Drive 2, Singapore 117546. 
}
$^{,}$\thanksA{Present Address: School of Data Science, City University of Hong Kong, Kowloon, Hong Kong SAR. Email: \texttt{yhuan26@cityu.edu.hk}}, 
Qiao Huang\thanksA{Group of Mathematical Physics (GFMUL), Department of Mathematics, Faculty of Sciences, University of Lisbon, Campo Grande, Edif\'{\i}cio C6, PT-1749-016 Lisboa, Portugal. 
}
$^{,}$\thanksA{Present Address: Division of Mathematical Sciences, School of Physical and Mathematical Sciences, Nanyang Technological University, 21 Nanyang Link, Singapore 637371. Email: \texttt{qiao.huang@ntu.edu.sg}}
$^{,}$\thanksB{Corresponding author},
Jinqiao Duan
\thanksA{The Dongguan Key Laboratory for Data Science and Intelligent Medicine, Department of Mathematics and Department of Physics, Great Bay University, Dongguan, Guangdong, 523000, China. Email: \texttt{duan@gbu.edu.cn}}
}
}

\date{}
\maketitle
\vspace{-0.3in}

\begin{abstract}
  The most probable transition paths of a stochastic dynamical system are the global minimizers of the Onsager--Machlup action functional 
  and can be described by a necessary but not sufficient condition, the Euler--Lagrange equation (a second-order differential equation with initial-terminal conditions) from a variational principle. 
  This work is devoted to showing a sufficient and necessary characterization for the most probable transition paths of stochastic dynamical systems with Brownian noise. We prove that, under appropriate conditions, the most probable transition paths are completely determined by a first-order ordinary differential equation. The equivalence is established by 
  showing that the Onsager--Machlup action functional of the original system can be derived from the corresponding  Markovian bridge process. 
  For linear stochastic systems and the nonlinear Hongler's model, the first-order differential equations determining the most probable transition paths are shown analytically to imply the Euler--Lagrange equations of the Onsager--Machlup functional. For general nonlinear systems, the determining first-order differential equations can be approximated, in a short time or for the small noise case. Some numerical experiments are presented to illustrate our results. 
\end{abstract}

\section{Introduction}

Stochastic differential equations (SDEs) have been widely used to describe complex phenomena in physical, biological, and engineering systems. Due to random fluctuations, transition phenomena between dynamically significant states occur in nonlinear systems. Hence, a practical issue is to capture the transition behavior between two metastable states and determine the \emph{most probable transition paths (MPTPs)} 
for the stochastic dynamical systems. This problem can be traced back to the derivation of the \emph{Onsager--Machlup action functional} (OM functional for short), which was firstly introduced by Onsager and Machlup to consider the probabilities of trajectories of thermodynamic states of an irreversible process, as a significant result of the theory of thermodynamic fluctuations \cite{OM53,MO53}. From the modern point of view, their works were devoted to SDEs with linear drifts and constant diffusion coefficients. A key step therein was to express the transition probability by Feynman's path integral. Later on, Tisza and Manning focused on the same issue for SDEs  
with nonlinear coefficients \cite{tisza1957fluctuations}. 

\newpage
Some authors followed the idea of \cite{OM53,MO53,tisza1957fluctuations} to seek the most probable transition paths of a diffusion process as the minimizers of the OM functional by a variational principle. This seems to be meaningless, since a solution curve of the Euler--Lagrange equation from the variational principle is necessarily twice differentiable, while trajectories of a diffusion process are almost surely nowhere differentiable and the probability of a single trajectory is zero. For these reasons, authors of \cite{Durr1978,Ikeda1980} sought the probability that a trajectory lies within a ``neighborhood'' of a differentiable curve. Such a neighborhood is a tube along the curve. They compared the probabilities of different tubes of the same ``thinness''. In a way, they interpreted the OM functional as the action giving the most probable infinitesimal tubes. We shall still use the terminology ``most probable transition paths'' as in those earlier references, to indicate solutions of the variational principle, even though they are not genuine trajectories. This viewpoint is the basis of many topics in physics and biochemistry, such as path integrals \cite{Kath1981path,Schulman1981,Wiegel1986,Wio2013,Khandekar2000}, trajectory sampling \cite{KA20,Lu2017}, transition-path theory \cite{EVE10,GLLL23}, transition rates \cite{Zuckerman2000,Gobbo2012}, transitions for non-Gaussian systems \cite{Chao2019,Huang2019}, biomolecular folding \cite{Faccioli2006,Wang2006}, information projection
\cite{selk2021information,bierkens2014explicit}, and Bayesian inverse problems \cite{dashti2013map,AKLS21}, etc. 

In this paper, we consider the following SDE in the state space $\mathbb{R}^k$:
\begin{equation}\label{firstequation}\left\{
  \begin{aligned}
    dX_t &= -\nabla U(X_t)dt+\sigma dW_t, \quad t>0, \\
    X_0 &=x_0,
  \end{aligned}\right.
\end{equation}
where $U:\mathbb{R}^k\rightarrow\mathbb{R}$ is a given potential function, $W = \{W_t\}_{t\ge0}$ is a standard $k$-dimensional Brownian motion, 
and $\sigma$ is a positive constant. The assumptions on coefficients will be given gradually in the sequel.
We  assume that $x_0\in\mathbb R^k$ is a \emph{metastable state} of (\ref{firstequation}), that is, $x_0$ is a stable state of the corresponding deterministic system $\dot x_t= -\nabla U(x_t)$ \cite{Huang2019}. 
For instance, when $k=1$ and $U = \frac{1}{4}x^4-\frac{1}{2}x^2$ is the Ginzburg--Landau double-well potential, the corresponding metastable states are $\pm 1$ (see Section \ref{examples} for more discussions about this example).

The original question of Onsager and Machlup can be presented in the language of dynamics as follows (see also \cite[Chapter 6, Section 9]{Ikeda1980}): 
Given the transition time $T$, among all differentiable curves connecting two metastable states $x_0$ and $x_T$ from $t=0$ to $t=T$, which one is ``most probable'' for the system \eqref{firstequation}? 
From the probabilistic perspective, \cite{Durr1978,Ikeda1980} translated the problem to studying the transitions of the system \eqref{firstequation} from $x_0$ to a neighborhood of $x_T$ at time $T$. 
See \cite{Zeitouni1989,Zeitouni1987,Zeitouni1988} for more developments. The common approach of them is to evaluate the measure of tubes around a  differentiable curve.
Their key observation is that (under some regularity assumptions) the probability of the solution process of $(\ref{firstequation})$ staying in a $\delta$-neighborhood (or equivalently, a $\delta$-tube) of a curve $\psi$ is given asymptotically by
\begin{equation}\label{OMapproximate}
\begin{split}
\mathbb{P}(\|X -\psi \|_T<\delta) \sim \exp(-S_X^{OM}(\psi))\mathbb{P}(\|W \|_T<\delta), \quad \delta\downarrow0,
\end{split}
\end{equation}
where $S_{X}^{OM}$ is called the \emph{Onsager--Machlup action functional} (OM functional) of \eqref{firstequation} and defined by
\begin{equation}\label{OM}
\begin{split}
S_{X}^{OM}(\psi)=\frac{1}{2}\int_0^T \left[ \frac{|\dot{\psi}(s) + \nabla U(\psi(s))|^2}{\sigma^2} - \triangle U(\psi(s)) \right]ds,
\end{split}
\end{equation}
and $\|\cdot\|_T$ denotes the uniform norm on the space $C([0,T],\mathbb R^k)$ of all continuous functions from $[0,T]$ to $\mathbb R^k$. 

The most probable transition paths of \eqref{firstequation} are determined by the minimizers (not necessarily unique) of $S_{X}^{OM}$ in a suitable space. Such a path can be considered as an infinite-dimensional version of the statistical mode, one of the numerical characteristics of probability laws \cite{AKLS21}. The study of the OM functional yielded many properties of most probable transition paths \cite{DLL21,Li2021}. From the standard calculus of variation, the OM functional can deduce an Euler--Lagrange (EL) equation, based on which a number of numerical methods have been proposed to seek for most probable transition paths \cite{E2002,E2004}. The EL equation has also been used to analyze the geometrical structure of most probable transition paths \cite{Soskin2006}. 

Notice that the Euler--Lagrange equation is a necessary but not sufficient description for the most probable transition paths, and it is a second-order equation with two boundary values which is in general hard to solve neither analytically nor numerically. There is a natural and crucial question: Can we have a sufficient and necessary, and simpler (at least analytically) description for the most probable transition paths? This is what motivates us in the present paper.

A similar notion for transition paths is the Markovian bridges. A Markovian bridge is obtained by conditioning a Markov process (in the sequel we always refer to the solution process $X$ of $(\ref{firstequation})$) to start from an initial state $x_0$ at time 0 and arrive at another state $x_T$ at time $T$. We call the resulting process the \emph{$(x_0,T,x_T)$-bridge derived from} $X$. It follows from the definition that the $(x_0,T,x_T)$-bridge has trajectories almost surely in the space $C_{x_0,x_T}[0,T]$ of all continuous $\mathbb R^k$-valued curves on time interval $[0,T]$ connecting $x_0$ and $x_T$.

In Markovian theory, the transition density function is usually assumed to be continuous in all of its variables. This implies that the path space $C_{x_0,x_T}[0,T]$ of the $(x_0,T,x_T)$-bridge is a null subset of the total space $C([0,T],\mathbb R^k)$, under the pushforward measure of $X$. Besides, we know that the measures on $C_{x_0,x_T}[0,T]$ induced by Markovian bridges are no longer quasi-translation invariant. For these reasons, the previous methods of deriving OM functionals in \cite{Durr1978,Ikeda1980,Zeitouni1989,Zeitouni1987,Zeitouni1988} are not applicable to Markovian bridges. Thus, based on these existing results, we do not know whether the most probable transition paths determined by the solution $X$ of \eqref{firstequation} and by its derived Markovian bridge coincide, even though it is intuitively true.

In the present paper, we will discuss the relations between \eqref{firstequation} and its derived $(x_0,T,x_T)$-bridge. These relations will help us to gain more insights in finding the most probable transition paths of system \eqref{firstequation}. We will show that the most probable transition paths of \eqref{firstequation} and the ones of its derived bridge coincide and are uniquely described by a first-order ordinary differential equation (ODE). For linear stochastic systems and nonlinear Hongler's model, the determining equations can be analytically verified to imply their EL equations. For general nonlinear cases, the determining equation can be approximated in a short time or for the small noise case. 

This paper is organized as follows. Section \ref{Mathematical introduction} serves for some preliminaries and definitions. We also obtain a connection between the most probable paths of SDEs and first-order ODEs. Section \ref{Markovianbridge} is devoted to the theory of Markovian bridges: We first study the finite-dimensional distributions of Markovian bridges in Subsection \ref{finited}; then we use SDE representations with only initial values to model Markovian bridges in Subsection \ref{SDEmarkov}; in Subsection \ref{OMbridge}, we show that OM functionals can be derived from bridge measures by a method different from \cite{Durr1978,Ikeda1980}; based on these, we obtain our main result, Theorem \ref{maintheory} in Subsection \ref{MPTPequivalent}, that the most probable transition paths of a stochastic dynamical system coincide with those of its derived Markovian bridge, and are determined by a first-order ordinary differential equation with an initial condition. The discussions of linear systems, nonlinear Hongler's model and more general nonlinear systems are presented in Section \ref{application}. In Section \ref{examples}, we present some examples to illustrate our results. We also discuss in Section \ref{sampling&control} the usage of SDE representations in trajectory sampling, and the optimal control formulation of most probable transition paths. Section \ref{conclusions} is reserved for a short summary.

\section{Diffusion processes, most probable transition paths and most probable paths}\label{Mathematical introduction}

In this section, we will make some preliminaries for the system \eqref{firstequation} and introduce the definition of most probable transition paths. We will also introduce and study the most probable paths of SDEs with general drifts.

\subsection{Preliminaries on diffusion measures and most probable transition paths}

We consider the following SDE on $\mathbb R^k$:
\begin{equation}\label{SDE}
dX_t= -\nabla U(X_t)dt+\sigma dW_t,
\end{equation}
where $W=(W^1,\cdots,W^k)$ is a standard $k$-dimensional Brownian motion.

We denote by $C[0,T]$ the space $C([0,T],\mathbb R^k)$ of all continuous functions from interval $[0,T]$ to $\mathbb R^k$, and equip it with the uniform topology induced by the following uniform norm so that it is a Banach space,
\begin{equation*}
\begin{split}
\|\phi\|_T=\sup_{t\in[0,T]}|\phi(t)|, \quad\psi\in C[0,T].
\end{split}
\end{equation*}
We also endow $C[0,T]$ with corresponding Borel $\sigma$-field $\mathcal B$ and canonical filtration $\{\mathcal B_t\}_{t\in[0,T]}$ given by
\begin{equation*}
  \mathcal B_t = \sigma\{ \omega(s) \mid \omega\in C[0,T], 0\le s\le t \}.
\end{equation*}

Let $\mathcal{A}$ be the infinitesimal generator of $X$, i.e.,
\begin{equation*}
\begin{split}
\mathcal{A}= -\nabla U\cdot\nabla+\frac{1}{2}\sigma^2\triangle,
\end{split}
\end{equation*}
We suppose the following:
\begin{assumption}\label{wellposedfirst}
  (1) The potential satisfies $U\in C^3(\mathbb{R}^k,\mathbb{R})$.\\
  (2) The local martingale problem associated with $\mathcal{A}$ is well-posed in $C[0,T]$ (e.g., \cite[Chapter 5, Definition 4.5]{Karatzas1991}). That is, for every $(s,x)\in[0,T)\times\mathbb{R}^k$, there exists a probability measure $P^{s,x}$ on $(C[0,T],\mathcal B)$, such that $P^{s,x}(\omega(r)=x,r\leq s)=1$, and for every $f\in C^2(\mathbb{R}^k)$, $M^f = \{M_t^f\}_{t\in[s,T]}$ is a local martingale with respect to the right continuous augmentation of $\{\mathcal B_t\}_{t\in[s,T]}$ under $P^{s,x}$, where
  $$M_t^f(\omega):=f(\omega(t))-f(\omega(s))-\int_s^t\mathcal{A} f(\omega(r))dr,\quad \omega\in C[0,T].$$
\end{assumption}

The regularity assumption for potential $U$ in \ref{wellposedfirst}-(1) can guarantee the pathwise uniqueness of \eqref{SDE} (e.g., \cite[Chapter IV, Theorem 3.1]{Ikeda1980}). On the other hand, the existence of solutions for the local martingale problem associated with $\mathcal{A}$ is equivalent to the existence of weak solutions for its corresponding SDE \eqref{SDE} \cite[Chapter 5, Corollary 4.8, Corollary 4.9]{Karatzas1991}, so that there exist a filtered probability space $(\Omega,\mathcal{F}, \{\mathcal{F}_t\}_{t\ge0}, \mathbb{P})$ satisfying usual conditions (e.g., \cite[Chapter 1, Definition 2.25]{Karatzas1991}), a $k$-dimensional Brownian motion $W$ on it and a continuous, adapted $\mathbb R^k$-valued process $X$ such that the equation \eqref{SDE} holds (in the sense of stochastic integral). Applying the result of Yamada--Watanabe \cite[Chapter 5, Corollary 3.23]{Karatzas1991}, we obtain the strong well-posedness of \eqref{SDE}. Moreover, Assumption \ref{wellposedfirst}-(2) implies that $X$ is strong Markov under $\mathbb{P}$ \cite[Theorem 4.4.2]{Ethier1986}. We denote the conditional probability measure $\mathbb{P}(\cdot\mid X_0=x_0)$ shortly by $\mathbb{P}^{x_0}(\cdot)$. Therefore, we have

\begin{lemma}[Well-posedness of SDEs]\label{strong-X}
 Under Assumption \ref{wellposedfirst}, the SDE \eqref{SDE} has a pathwisely unique strong solution $X=\{X_t\}_{t\in[0,T]}$ on the probability space $(\Omega,\mathcal{F}, \{\mathcal{F}_t\}_{t\ge0}, \mathbb{P}^{x_0})$, which is strong Markov and satisfies $\mathbb{P}^{x_0}(X_0 = x_0) = 1$.
\end{lemma}

The transition semigroup $(T_{s,t})_{0\leq s<t}$ of the system \eqref{SDE} is defined as
\begin{equation*}
\begin{split}
(T_{s,t}f)(x)=\mathbb{E}(f(X_t)|~X_s=x),
\end{split}
\end{equation*}
for each $f\in \mathfrak{B}_b(\mathbb{R}^k)$ and $x\in\mathbb{R}^k$, here $\mathfrak{B}_b(\mathbb{R}^k)$ denotes the space of all measurable and bounded functions $f:\mathbb{R}^k\rightarrow \mathbb{R}^k$. The notation $\mathbb{E}(\cdot \mid X_s=x)$ denotes the expectation with respect to the regular conditional probability measure $\mathbb{P}(\cdot \mid X_s=x)$. We suppose that $T_{s,t}$ admits a transition density $p(\cdot,t|x,s)$ with respect to a $\sigma$-finite measure $\nu$ on $\mathbb{R}^k$, in the sense that
\begin{equation*}
T_{s,t}f(x)=\int_{\mathbb{R}^k} f(y)p(y,t|x,s)\nu(dy).
\end{equation*}
For simplicity, we assume that $\nu$ is the Lebesgue measure. Since the drift term $-\nabla U$ and diffusion coefficient $\sigma$ do not depend on time, we know that this transition density is time homogenous \cite{Ikeda1980}, i.e.,
\begin{equation*}
p(y,t+s|z,t)=p(y,s|z,0),
\end{equation*}
for every $t,s\in(0,\infty)$ and $y,z\in \mathbb{R}^k$.

Under Assumption \ref{wellposedfirst}, the transition densities satisfy the following properties (see Chapter 6 of \cite{Bogachev2015}):\\
(\lowercase \expandafter {\romannumeral 1}) $(s,y,x)\mapsto p(y,s|x,0)$ is joint continuous,\\
(\lowercase \expandafter {\romannumeral 2}) The transition density function satisfies the Kolmogorov forward equation (or Fokker-Planck equation)
\begin{equation}\label{forward}
\begin{split}
\frac{\partial p(x,t|x_0,0)}{\partial t}= \nabla(\nabla U(x)p(x,t|x_0,0))+\frac{1}{2}\sigma^2\triangle p(x,t|x_0,0),
\end{split}
\end{equation}
and the Kolmogorov backward equation
\begin{equation}\label{backward}
\begin{split}
\frac{\partial p(x_T,T|x,t)}{\partial t}= \nabla U(x)\cdot\nabla p(x_T,T|x,t)-\frac{1}{2}\sigma^2\triangle p(x_T,T|x,t),
\end{split}
\end{equation}
both in the sense of generalized functions.

Due to the strong Markov property of the process $X$, the Chapman--Kolmogorov equations
\begin{equation*}
\begin{split}
p(y,T|x_0,0)=\int_{\mathbb{R}^k}p(z,T-s|x_0,0)p(y,s|z,0)dz,
\end{split}
\end{equation*}
hold for all $y\in \{x\mid p(x,T|~x_0,0)>0\}$ and all $0<s<T$.

We suppose that $x_0$ is a fixed state of system \eqref{SDE}. And let $x_T$ denote another given state of system \eqref{SDE}. The space of trajectories of $X$ is the space $C_{x_0}[0,T]$ of continuous functions
\begin{equation*}
\begin{split}
C_{x_0}[0,T]=\{\phi \mid \phi:[0,T]\rightarrow\mathbb{R}^k~\mbox{is continuous},~\phi(0)=x_0\}.
\end{split}
\end{equation*}
Note that $C_{x_0}[0,T]$ is not a linear space unless $x_0=0$. We equip $C_{x_0}[0,T]$ with the subspace topology of $C[0,T]$,
and denote the corresponding Borel $\sigma$-field by $\mathcal{B}^{x_0}_{[0,T]}$. There is another way to realize elements in $\mathcal{B}^{x_0}_{[0,T]}$, in terms of cylinder sets, instead of open sets. A cylinder set of $C_{x_0}[0,T]$ is of the form
\begin{equation*}
\begin{split}
I=\{\phi\in C_{x_0}[0,T]\mid \phi(t_1)\in E_1,\cdots,\phi(t_n)\in E_n\},
\end{split}
\end{equation*}
where $0\le t_1<\cdots<t_n\leq l$ and $E_i$'s are Borel sets of $\mathbb{R}^k$. It is well known \cite{Karatzas1991} that $\mathcal{B}^{x_0}_{[0,T]}$ is the $\sigma$-field generated by all cylinder sets, that is, the smallest $\sigma$-field containing all cylinder sets. An open tube set $K_T(\psi,\delta)$ is defined as
\begin{equation*}
\begin{split}
K_T(\psi,\delta)=\{\phi\in C_{x_0}[0,T] \mid \|\psi-\phi\|_T<\delta\},
\end{split}
\end{equation*}
where $\delta>0$ is called the tube size. The corresponding closed tube set is
\begin{equation*}
\begin{split}
\bar{K}_T(\psi,\delta)=\{\phi\in C_{x_0}[0,T] \mid \|\psi-\phi\|_T\leq\delta\},
\end{split}
\end{equation*}
which is the closure of $K_T(\psi,\delta)$ under the uniform topology. Let $B_\rho(x)$ denote the open ball centered at $x\in \mathbb{R}^k$ with radius $\rho>0$, we denote by $\bar{B}_\rho(x)$ the corresponding closed ball.

The measure $\mu_X^{x_0}$ induced by $X$ on the space $(C_{x_0}[0,T], \mathcal{B}^{x_0}_{[0,T]})$ is defined by
\begin{equation*}
\begin{split}
\mu_X^{x_0}(B)&=\mathbb{P}^{x_0}(\{\omega\in\Omega \mid X(\omega)\in B\}),
\end{split}
\end{equation*}
for all $B\in\mathcal{B}^{x_0}_{[0,T]}$. Recall that the measure $\mu_{\sigma W}^{x_0}$ induced by the Brownian motion $\sigma W$ is called the \emph{Wiener measure}.  Once a positive $\delta$ is given, we can compare the probabilities of closed tubes for all $\psi\in C_{x_0}[0,T]$ using $\mu_X^{x_0}(\bar{K}_T(\psi,\delta))$. This enables us to discuss the problem of finding the  \emph{most probable transition path(s)} of $X$.

In general, we have the following definition for most probable transition paths (MPTPs).

\begin{definition}[MPTPs]
The most probable transition path of the system \eqref{SDE} connecting two given states $x_0$ and $x_T$, is a path $\psi^*$ that makes the OM functional \eqref{OM} achieve its minimum value in the following path space,
\begin{equation*}
\begin{split}
C^2_{x_0,x_T}[0,T] := \{ \psi:[0,T]\rightarrow\mathbb{R}^k \mid \dot{\psi},~\ddot{\psi}~\mbox{exist and are continuous},~\psi(0)=x_0,~\psi(T)=x_T\},
\end{split}
\end{equation*}
that is, $\psi^* \in  C^2_{x_0,x_T}[0,T]$ is such that
\begin{equation*}
\begin{split}
S_{X}^{OM}(\psi^*)=\inf_{\psi\in C^2_{x_0,x_T}[0,T]}S_{X}^{OM}(\psi).
\end{split}
\end{equation*}
\end{definition}

By \eqref{OMapproximate}, this is equivalent to say that for all $\psi\in C^2_{x_0,x_T}[0,T]$,
\begin{equation}\label{MPTP-2}
\begin{split}
\lim_{\delta\downarrow0}\frac{\mu^{x_0}_X(K_T(\psi^*,\delta))}{\mu_X^{x_0}(K_T(\psi,\delta))}\geq1.
\end{split}
\end{equation}
One can replace the open tubes in \eqref{MPTP-2} by closed tube sets, adopting a slight modification of the proof of \cite[Theorem 9.1]{Ikeda1980} (or \cite[Section 4]{Durr1978}, \cite[Theorem 1]{Zeitouni1988}).

\subsection{Most probable paths of SDEs}

In the following lemma we show that under a given probability measure on path space, the probability of a closed tube set can be approximated by the probabilities of a family of cylinder sets. This property helps us to study the tube probability easily.

\begin{lemma}[Approximation for probabilities of closed tubes]\label{cylinderset}
  Let $\mu$ be a probability measure defined on $(C_{x_0}[0,T],\mathcal{B}^{x_0}_{[0,T]})$. For each closed tube set $\bar{K}_T(\psi,\delta)$ with $\psi\in C_{x_0}[0,T]$ and $\delta>0$, there exists a family of closed cylinder sets $\{\bar{I}_n\}_{n=1}^\infty$ such that
\begin{equation*}
\begin{split}
\mu (\bar{K}_T(\psi,\delta))= \lim_{n\rightarrow\infty}\mu (\bar{I}_n).
\end{split}
\end{equation*}
\end{lemma}
\begin{proof}
The proof is separated into two steps.


\emph{Step 1.} Let $\mathbb{Q}$ denote the countable set of rational numbers in $\mathbb{R}$. Since $(0,T)\cap \mathbb{Q}$ is a countable set, we denote it as a sequence $\{q_1,q_2,\cdots,q_n,\cdots\}$. 
Define a family of incremental sequences $\{Q_n\}_{n=1}^\infty$ by
\begin{equation*}
Q_n:=\{q_1,\cdots,q_n\}.
\end{equation*}
Then we have $(0,T)\cap \mathbb{Q} = \cup_{n=1}^\infty Q_n$. By the continuity, we can derive the following equalities:
\begin{equation}\label{est-3}
\begin{split}
&\left\{ w\in C_{x_0}[0,T] \Bigg| \sup_{t\in[0,T]} |w(t)-\psi(t)|\leq\delta \right\}\\
=& \left\{ w\in C_{x_0}[0,T] \Bigg| \sup_{t\in(0,T)\cap \mathbb{Q}}|w(t)-\psi(t)|\leq\delta \right\} \\
=& \bigcap_{t\in(0,T)\cap \mathbb{Q}} \left\{ w\in C_{x_0}[0,T] \mid |w(t)-\psi(t)|\leq\delta \right\} \\
=& \bigcap_{n=1}^\infty \bigcap_{t\in Q_n} \left\{ w\in C_{x_0}[0,T] \mid |w(t)-\psi(t)|\leq\delta \right\} \\
=& \bigcap_{n=1}^\infty \left\{ w\in C_{x_0}[0,T] \Bigg| |w(t)-\psi(t)|\leq\delta, \forall t\in Q_n \right\}.
\end{split}
\end{equation}

\emph{Step 2.} Noting that the family $\{ \{w\in C_{x_0}[0,T] \mid |w(t)-\psi(t)|\leq\delta, \forall t\in Q_n\}: n=1,\cdots,\infty\}$ is decreasing since $\{Q_n\}_{n=1}^\infty$ is increasing, we have
\begin{equation*}
\begin{split}
\mu(\bar{K}_T(\psi,\delta))
=&\ \mu\left(\left\{w\in C_{x_0}[0,T] \Bigg| \sup_{t\in[0,T]} |w(t)-\psi(t)|\leq\delta\right\}\right)\\
=&\ \mu\left(\bigcap_{n=1}^\infty \{w\in C_{x_0}[0,T] \mid |w(t)-\psi(t)|\leq\delta, \forall t\in Q_n\}\right)\\
=&\ \lim_{n\rightarrow\infty}\mu(\{w\in C_{x_0}[0,T] \mid |w(t)-\psi(t)|\leq\delta, \forall t\in Q_n\})\\
=&\ \lim_{n\rightarrow\infty}\mu(\bar{I}_n(\psi,\delta)),
\end{split}
\end{equation*}
where $\bar{I}_n(\psi,\delta):=\{\phi\in C_{x_0}[0,T]\mid \phi(t)\in \bar{B}_\delta(\psi(t)), \forall t\in Q_n\}$
is a closed cylinder set. The proof is complete.
\end{proof}

\begin{remark}
  In general, this lemma does not work for open tubes. Indeed, if we replace the closed tubes $\bar K_T$ and closed cylinder sets $\bar I_n$ by their open versions, then the first two equalities of \eqref{est-3} should read
\begin{equation*}
\begin{split}
\left\{ w\in C_{x_0}[0,T] \Bigg| \sup_{t\in[0,T]} |w(t)-\psi(t)| < \delta \right\} &= \left\{ w\in C_{x_0}[0,T] \Bigg| \sup_{t\in(0,T)\cap \mathbb{Q}}|w(t)-\psi(t)| < \delta \right\} \\
&\subset \bigcap_{t\in(0,T)\cap \mathbb{Q}} \left\{ w\in C_{x_0}[0,T] \mid |w(t)-\psi(t)| < \delta \right\}.
\end{split}
\end{equation*}
\end{remark}

Now we consider the following SDE on $\R^k$ with general drift:
    \begin{equation}\label{SDEt}
dZ_t= f(t,Z_t)dt+\sigma dW_t,\quad Z_0=x_0;
\end{equation}
and its deterministic counterpart, i.e., the following (first-order) ODE getting rid of the noise term:
 \begin{equation}\label{MPPequ}
d\psi(t)= f(t,\psi(t))dt,\quad \psi(0)=x_0,
\end{equation}
where $f$ is an $\R^k$-valued function defined on $[0,T]\times \R^k$.

Inspired by \eqref{MPTP-2}, we have the following definition of most probable paths (MPPs). The difference between most probable transition paths and most probable paths is that, MPTPs are determined by conditioning on starting points and ending points while MPPs only specify starting points.
\begin{definition}[MPPs]
The most probable path (MPP) of the system  \eqref{SDEt} is a path $\psi^*\in C^2_{x_0}[0,T]$ such that
\begin{equation}\label{eqn-1}
\begin{split}
\lim_{\delta\downarrow0}\frac{\mu_Z^{x_0}(\bar{K}_T(\psi^*,\delta))}{\mu_Z^{x_0}(\bar{K}_T(\psi,\delta))}\geq1, \quad \text{for all } \psi \in C^2_{x_0}[0,T].
\end{split}
\end{equation}
\end{definition}

The following elementary lemma gives a characterization of MPPs. This will play a crucial role when studying MPTPs of Markovian bridges derived from \eqref{SDE} in the next section.

\begin{lemma}[MPPs \& 1st order ODEs]\label{MPP}
Suppose that the well-posedness of weak solutions of \eqref{SDEt} on $[0,T]$ is promised. Then a path $\psi^*\in C^2_{x_0}[0,T]$ is a most probable path of \eqref{SDEt} if and only if it solves the ODE \eqref{MPPequ}.
\end{lemma}
\begin{proof}
Without loss of generality, we assume that $T$ is rational. 
Due to Lemma \ref{cylinderset}, we know that there exist cylinder sets $\bar I_n(\psi^*,\delta)$ and $\bar I_n(\psi,\delta)$ whose probabilities under $\mu_Z^{x_0}$ approach to $\mu_Z^{x_0}(\bar{K}_T(\psi^*,\delta))$ and $\mu_Z^{x_0}(\bar{K}_T(\psi,\delta))$ respectively. Then the definition \eqref{eqn-1} reduce to
\begin{equation*}
\begin{split}
\lim_{\delta\downarrow0} \lim_{n\rightarrow\infty}\frac{\mu_Z^{x_0}(\bar{I}_n(\psi^*,\delta))}{\mu_Z^{x_0}(\bar{I}_n(\psi,\delta))}.
\end{split}
\end{equation*}

For SDE \eqref{SDEt}, one typically applies the Euler--Maruyama algorithm as an
approximate method for propagating the position as a function
of time. That is,
\begin{equation*}
    z_{i+1}=z_i+f(t_i,z_i)\Delta t+\sigma \sqrt{\Delta t} \xi_i,\quad 0=t_0<t_1<\cdots<t_m = T,\quad i=1,\cdots,m, 
\end{equation*}
where $\{\xi_i\}_{i=0}^{m-1}$ is a family of $k$-dimensional independent standard Gaussian random variables with zero mean and unit variance, $\Delta t = T/m = t_{i+1} - t_i$ is the time step. We denote the change of variables from $\{\xi_i\}_{i=0}^{m-1}$ to $\{z_i\}_{i=1}^m$ as a map $\mathfrak{T}_m$: $\mathbb{R}^{mk}\rightarrow\mathbb{R}^{mk}$, $\{\xi_i\}\mapsto\mathfrak{T}(\{\xi_i\}):=\{z_{i+1}\}$. It is easy to check that the Jacobian determinant of this change of variables is
\begin{equation*}
    \frac{\partial (\xi_0,\cdots,  \xi_{m-1})}{\partial (z_1,\cdots,  z_m)}=\left(\frac{1}{\sigma\sqrt{\Delta t}}\right)^m,
\end{equation*}
and thus the map $\mathfrak{T}_m$ is one-to-one.
And thus for a subset $Q_n$ (with $n\ge m$) of rational numbers containing $\{t_i\}_{i=1}^m$, and for any $\psi$ in $C^2_{x_0}[0,T]$, we have
\begin{align}
  &\ \mu_Z^{x_0}(\bar{I}_n(\psi,\delta)) \widesim{m,n\to\infty}
  \mathbb P^{x_0} \left( z_i\in \bar B_\delta(\psi(t_i)), \forall i=1,\cdots, m \right) \notag \\
    =&\ \mathbb P^{x_0} \left( (\xi_0,\cdots, \xi_{m-1}) \in \mathfrak{T}_m^{-1} \left(\prod_{i=1}^m \bar B_\delta(\psi(t_i)\right) \right) \notag \\ 
    =&\ \int_{\mathbb{R}^k}\cdots\int_{\mathbb{R}^k}\ind_{\mathfrak{T}_m^{-1}\left(\prod_{i=1}^m \bar B_\delta(\psi(t_i) \right)} \left(\frac{1}{\sqrt{(2\pi)^k}}\right)^m\exp \left\{ - \sum_{i=0}^{m-1}\frac{|\xi_i|^2}{2} \right\}d\xi_0\cdots d\xi_{m-1} \label{eqn-a} \\
    =&\ \int_{\bar B_\delta(\psi(t_m))}\cdots\int_{\bar B_\delta(\psi(t_1))} \left(\frac{1}{\sqrt{(2\pi)^k\Delta t}\sigma}\right)^m\exp \left\{- \frac{\Delta t}{\sigma^2} \underbrace{\sum_{i=0}^{m-1}\frac{1}{2}\left|  \frac{z_{i+1}-z_i}{\Delta t} -f(t_i,z_i)\right|^2}_{=:\Xi} \right\}dz_1\cdots dz_m \label{eqn-b} \\
    =:&\ J(\psi,\delta), \notag
\end{align}
where $\ind_{E}$ denotes the indicator of set $E$, the first row means that the two quantities on both sides of $\sim$ are arbitrarily close when $m,n\to\infty$. 

Now let $\delta$ go to zero. If $\psi^*$ solves \eqref{MPPequ}, then the term $\Xi$ in the exponential of  $J(\psi,\delta)$ is uniformly small for all sufficiently small $\delta$, say $\Xi <\e$, so that $$J(\psi^*,\delta)> \left(\frac{|B_\delta(0)|}{\sqrt{2\Delta t\pi}\sigma}\right)^m e^{-\e}.$$
Whereas if $\psi$ does not satisfy \eqref{MPPequ}, then there exists a $\hat t\in [0,T)$ such that $\dot\psi(\hat t)\ne f(\hat t,\psi(\hat t))$. In this case we can choose $m,n$ sufficiently large and $\delta$ sufficiently small, so that there exists an $0\le i\le m-1$ such that $\hat t\in [t_i,t_{i+1}]$ and
$$\left|  \frac{z_{i+1}-z_i}{\Delta t} -f(t_i,z_i)\right|^2\ge \e, \quad \forall (z_i,z_{i+1}) \in \bar B_\delta(\psi(t_i)) \times \bar B_\delta(\psi(t_i+1)),$$ 
which implies 
$$J(\psi,\delta)\le \left(\frac{|B_\delta(0)|}{\sqrt{2\Delta t\pi}\sigma}\right)^m e^{-\e}.$$
This completes the proof.


\end{proof}

\begin{remark}\label{rmk-a}
(i). As a consequence, under the assumption of the above lemma, the existence and uniqueness of a most probable path of \eqref{SDEt} is equivalent to the existence and uniqueness of a solution of \eqref{MPPequ} in $C^2_{x_0}[0,T]$.


(ii). If the drift function $f$ is only defined on $[0,T)\times\R^k$ which is the case when we study the MPPs of bridge processes, solutions of \eqref{SDEt} and \eqref{MPPequ} on $[0,T]$ need to be understood as (almost surely) continuous extension from $[0,T)$ to $[0,T]$.

(iii). The author of \cite{Dekker78} proved the same result as the above lemma by minimizing the discretized OM functional. 

(iv). The authors of \cite{MP16} wrote the relations \eqref{eqn-b} and \eqref{eqn-a} as the following formal expressions:
\begin{equation}\label{Pp}
  -\ln \mathbb{P}^{x_0}=C + \frac{\Delta t}{\sigma^2}\sum_i \frac{1}{2}\left|  \frac{z_{i+1}-z_i}{\Delta t} + f(t_i, z_i) \right|^2, \quad \text{and} \quad  \mathbb{P}^{x_0}\propto \prod_i \exp\left(-\frac{1}{2}\xi_i^2\right).
\end{equation}
where $C$ is a constant. These expressions are a bit misleading because they do not indicate the sets where the probability measure $\mathbb{P}^{x_0}$ take values or its probability density is integrated.
The authors claimed from the second expression of \eqref{Pp} that the ``trajectory probability'' is completely determined by the noise and thus is independent of the system under study. This is inappropriate, since the set under the integral in the complete expression \eqref{eqn-a}, i.e., $\ind_{\mathfrak{T}_m^{-1}\left(\prod_{i=1}^m \bar B_\delta(\psi(t_i) \right)}$, does depend on the underlying system (recall the definition of $\mathfrak{T}_m$ in the proof).
\end{remark}

\section{Markovian bridges and most probable transition paths}\label{Markovianbridge}

In this section, we will systematically study the relations between Markovian bridges and most probable transition paths. 

\subsection{Finite-dimensional distributions of Markovian bridges}\label{finited}

Lemma \ref{cylinderset} enables us to use cylinder sets to approximate the tube probabilities. So it is essential for us to consider the finite-dimensional distributions of Markovian bridges.

Recall that we have fixed an $x_0\in\mathbb R^k$. Under our setting, we know that the conditional probability distribution $\mathbb{P}^{x_0}(X\in\cdot \mid X_T)$ has a regular version, that is, it determines a regular conditional distribution of $X$ given $x_T$ under $\mathbb{P}^{x_0}$ \cite{Chaumont2011,Fitzsimmons1993}. We denote by $\mu_{X}^{x_0,\cdot}$ the corresponding probability kernel from $\mathbb R^k$ to $C_{x_0}[0,T]$, and call it the \emph{bridge measure}. This means that for $(\mathbb P^{x_0}\circ X_T^{-1})$-a.s. $x_T\in\mathbb R^k$ and all $B\in \mathcal{B}^{x_0}_{[0,T]}$,
\begin{equation*}
  \mu_{X}^{x_0, x_T}(B) = \mathbb{P}^{x_0}(X\in B \mid X_T = x_T). 
\end{equation*}

Under $\mathbb{P}^{x_0}(\cdot\mid X_T=x_T)$, the process $\{X_t\}_{0\leq t<T}$ is the $(x_0,T,x_T)$-bridge derived from $X$. And this bridge is still strong Markovian \cite{Chaumont2011,Fitzsimmons1993}, with transition densities
\begin{equation}\label{conditionalp}
\begin{split}
p^{x_0,x_T}(y,t|x,s)&=\frac{p(y,t-s|x,0)p(x_T,T-t|y,0)}{p(x_T,T-s|x,0)}\\
&=\frac{p(y,t|x,s)p(x_T,T|y,t)}{p(x_T,T|x,s)},~~0\leq s<t<T.
\end{split}
\end{equation}
Moreover $\mu_{X}^{x_0,x_T}(\{\psi\in C_{x_0}[0,T]\mid \psi(T)=x_T\})=1$.

For a cylinder set $I=\{\psi\in C_{x_0}[0,T]\mid\psi(t_1)\in E_1,\cdots,\psi(t_n)\in E_n\}$ with $0< t_1<t_2<\cdots<t_n<T$ and $E_i$'s being Borel sets of $\mathbb{R}^k$, we have that
\begin{equation}\label{finitedistribution}
\begin{split}
\mu_X^{x_0,x_T}(I)
=& \int_{\{x_i\in E_i,i=1,\cdots,n\}}p^{x_0,x_T}(x_1,t_1|x_0,0)\cdots p^{x_0,x_T}(x_n,t_n|x_{n-1},t_{n-1})dx_1\cdots dx_n\\
=& \int_{\{x_i\in E_i,i=1,\cdots,n\}}\frac{p(x_1,t_1|x_0,0)p(x_T,T|x_1,t_1)}{p(x_T,T|x_0,0)}\frac{p(x_2,t_2|x_1,t_1)p(x_T,T|x_2,t_2)}{p(x_T,T|x_1,t_1)}\\
& \qquad\qquad\qquad\quad\cdots \frac{p(x_{n},t_{n}|x_{n-1},t_{n-1})p(x_T,T|x_{n},t_{n})}{p(x_T,T|x_{n-1},t_{n-1})}dx_1\cdots dx_{n}\\
=&\frac{1}{p(x_T,T|x_0,0)}\int_{\{x_i\in E_i,i=1,\cdots,n\}}p(x_1,t_1|x_0,0)p(x_2,t_2|x_1,t_1)\cdots p(x_T,T|x_{n},t_{n})dx_1\cdots dx_{n}.
\end{split}
\end{equation}

\subsection{SDE representation for Markovian bridges}\label{SDEmarkov}
In this subsection, we represent Markovian bridges via SDEs only with initial values.

Combining equations $(\ref{forward})$, $(\ref{backward})$ and $(\ref{conditionalp})$, the transition probability density function $p^{x_0,x_T}(x,t|x_0,0)$ satisfies the following partial differential equation \cite{Cetin2016}:
\begin{equation*}
\begin{split}
\frac{\partial p^{x_0,x_T}(x,t|x_0,0)}{\partial t}= \nabla\left[ (\nabla U(x) - \sigma^2\nabla \ln p(x_T,T|x,t))p^{x_0,x_T}(x,t|x_0,0) \right]+\frac{1}{2}\sigma^2\triangle p^{x_0,x_T}(x,t|x_0,0).
\end{split}
\end{equation*}
Here and after, the spatial differentiation of $\ln p$ (or $p$) is always implemented with respect to the third variable of $p$, e.g., $\nabla_x \ln p(\cdot, \cdot| x,\cdot)$. Formally, this equation has the form of a Fokker-Planck equation. Thus we can associate to the transition density $p^{x_0,x_T}(x,t|x_0,0)$ a new $k$-dimensional SDE on a certain probability space $(\tilde{\Omega},\mathcal{\tilde F}, \{\tilde{\mathcal{F}}_t\}_{t\ge0}, \mathbb{\tilde P})$ that is rich enough to support a standard $k$-dimensional Brownian motion $\tilde{W}=(\tilde{W}^1,\cdots,\tilde{W}^k)$:
\begin{equation}\label{newsde}
   dY_t=\left[ - \nabla U(Y_t)+\sigma^2\nabla\ln p(x_T,T|Y_t,t) \right]dt+\sigma d\tilde{W}_t,\quad t\in[0,T).
\end{equation}
We refer to \eqref{newsde} as the \emph{bridge SDE} associated with \eqref{SDE}.
This equation was originally obtained by Doob \cite{Doob1957} from the probabilistic point of view and
is known as the \emph{Doob h-transform} of SDE \eqref{SDE}. For the sake of notational simplicity, we introduce the drift of the bridge SDE \eqref{newsde} by
\begin{equation}\label{m-drift}
b(t,x) := - \nabla U(x)+\sigma^2\nabla\ln p(x_T,T|x,t),
\end{equation}
and call it the modified drift.

The existence and uniqueness of weak and strong solutions of (\ref{newsde}) were established in \cite{Cetin2016} under some mild assumptions. Specifically, under Assumption \ref{wellposedfirst} (and the following Assumption \ref{assumptiononEU} if $k\geq2$), the existence and uniqueness of the strong solution of $(\ref{newsde})$ are promised (\cite[Theorem 4.1]{Cetin2016}). Denote, similar as before, by $\mathbb{\tilde P}^{x_0}$ the conditional probability $\mathbb{\tilde P}(\cdot\mid Y_0=x_0)$. Then, it concluded that for each Borel set $E$ of $\mathbb{R}^k$,
\begin{equation}\label{densityY}
\begin{split}
\mathbb{\tilde P}^{x_0}(Y_t\in E)=\int_E\frac{p(y,t|x_0,0)p(x_T,T|y,t)}{p(x_T,T|x_0,0)}dy, \quad 0<t<T,
\end{split}
\end{equation}
and automatically $\mathbb{\tilde P}^{x_0}(Y_T=x_T)=1$.

\begin{assumption}\label{assumptiononEU}
  When $k\geq2$, we assume in addition that $p(x_T,T|x_0,0)>0$ and $-\triangle U \ge \xi$, where $\xi\in\mathbb R$ is a constant, and $h\in C^{1,2}([0,T)\times\mathbb R^k)$ where $h(t,x)=p(x_T,T|x,t)$.
\end{assumption}

\begin{lemma}[Well-posedness of bridge SDEs]\label{strong-Y}
 Under Assumptions \ref{wellposedfirst} and \ref{assumptiononEU}, the SDE \eqref{newsde} has a pathwisely unique strong solution $Y=\{Y_t\}_{t\in[0,T]}$ on the probability space $(\tilde{\Omega},\mathcal{\tilde F}, \{\tilde{\mathcal{F}}_t\}_{t\ge0}, \mathbb{\tilde P}^{x_0})$, which is strong Markov and satisfies $\tilde{\mathbb{P}}^{x_0}(Y_0 = x_0, Y_T=x_T) = 1$.
\end{lemma}

\begin{proof}
The well-posedness of the system \eqref{newsde} in the case $k=1$ has been proved in \cite[Example 2.2]{Cetin2016}, and it can be verified that the conditions therein are all fulfilled by our Assumption \ref{wellposedfirst} (cf. \cite[Proposition 4.1]{Cetin2016}). Thus in the one-dimensional case, the existence and uniqueness of the strong solutions to the Markovian bridge systems are promised in our framework. If $k\geq2$ and $- \triangle U\geq \xi$, according to \cite[Theorem 1]{Metafune2011} we know that, there exists a positive constant $M_0$ depending only on the dimension $k$ and the diffusion coefficient $\sigma$ such that, if $\xi\neq0$,
\begin{equation*}
0\leq p(y,t|x,0)\leq M_0e^{-\frac{\xi}{2}t}\left(\frac{k}{|\xi|}\right)^{-k/2}\left(\cosh\left(-\frac{\xi t}{k}-1\right)\right)^{-k/4},
\end{equation*}
and if $\xi=0$,
\begin{equation*}
0\leq p(y,t|x,0)\leq M_0t^{-k/2}.
\end{equation*}
These estimates hold for all $t>0$ and $x,y\in \mathbb{R}^k$. And these estimates together with Assumption \ref{wellposedfirst}, $h\in C^{1,2}([0,T),\mathbb R^k)$ and $p(x_T,T|x_0,0)>0$ ensure the assumptions in \cite[Theorem 4.1]{Cetin2016} are fulfilled, thus the existence and uniqueness of the strong solution of $(\ref{newsde})$ are promised. The strong Markovian property follows from \cite[Corollary 4.1]{Cetin2016}.
\end{proof}

\begin{remark}\label{remark-1}
(i). The configurations $(\tilde{\Omega},\mathcal{\tilde F}, \{\tilde{\mathcal{F}}_t\}_{t\ge0}, \mathbb{\tilde P})$, $\tilde W$ of bridge SDE \eqref{newsde} do \emph{not} need to coincide with the configuration $(\Omega,\mathcal{F}, \{\mathcal{F}_t\}_{t\ge0}, \mathbb{P})$, $W$ of original SDE \eqref{SDE}, since the ways of obtaining their strong well-posedness in Lemmas \ref{strong-X} and \ref{strong-Y} are both via the Yamada--Watanabe argument, which means that as soon as the probability space is rich enough to support a Brownian motion, the SDE can have a unique strong solution on it, cf. \cite[Section 5.3]{Karatzas1991}.

(ii). Comparing bridge SDE \eqref{newsde} and original SDE \eqref{SDE}, it can be seen that it is the extra potential $\sigma^2\ln p(x_T,T|x,t)$ that ``force'' the process $Y$ to reach $x_T$ at time $t=T$. The physical realization for \eqref{newsde} is to add a virtual potential to the original non-equilibrium thermodynamical system,
to force (almost) all trajectories of the process to pass from an initial metastable state to a particular final metastable state.
\end{remark}

\subsection{Onsager--Machlup functional and bridge measures}\label{OMbridge}

In this subsection, we prove an essential result that is important for our main theorem in the next subsection, we summarize it as the following lemma:
\begin{lemma}[OM functional \& bridge measures]\label{OMderived}
There exists a constant $C>0$, such that for each $\psi\in C^2_{x_0,x_T}[0,T]$,
\begin{equation*}
\mu_X^{x_0,x_T}(\bar{K}_T(\psi,\delta))\sim C\exp\left( -S^{OM}_X(\psi) \right)\ \mu_{\sigma W}^{0,0}(\bar{K}_T(0,\delta)) \quad \text{as } \delta\downarrow0.
\end{equation*}
\end{lemma}

\begin{remark}
  This lemma will be proved by adopting Lemma \ref{cylinderset}, thus the result holds only for closed tubes but not for open tubes.
\end{remark}

Conditioning on the event that the solution process $X$ of \eqref{SDE} hits the point $x_T$ at time $T$, the regular conditional probability measure (i.e., the bridge measure) $\mu_X^{x_0,x_T}$ obeys the following informal stochastic boundary value problem, called \emph{conditioned SDE} \cite{Pinski2012},
\begin{equation}\label{sdetwobridge}\left\{
\begin{aligned}
dX_t&=  - \nabla U(X_t)dt+\sigma dW_t,\\
X_0&=x_0,\quad X_T=x_T,
\end{aligned}\right.
\end{equation}
which should still be understood as the original SDE \eqref{SDE} under the conditional probability $\mathbb{P}^{x_0}(\cdot| x_T =x_T)$.
A significant difference between the bridge SDE \eqref{newsde} and the conditioned SDE \eqref{sdetwobridge} is that the former is only conditioned on initial value while the latter is conditioned on initial and terminal boundary values.

Now, the measure $\mu_X^{x_0,x_T}$ can be characterized via its density with respect to the Brownian bridge measure $\mu_{\sigma W}^{x_0,x_T}$ corresponding to the case $\nabla U\equiv0$. To see this, for the unconditioned process $X$ in \eqref{SDE}, the Girsanov formula gives
\begin{equation}\label{Girsanov}
\begin{split}
\frac{d\mu_X^{x_0}}{d\mu_{\sigma W}^{x_0}}(x)=\exp\left\{ - \int_0^T\frac{\nabla U(x(t))}{\sigma}dx(t)-\frac{1}{2}\int_0^T\frac{|\nabla U(x(t))|^2}{\sigma^2}dt \right\}.
\end{split}
\end{equation}
This expression contains a stochastic integral term. Using It\^{o}'s formula we obtain that
\begin{equation*}
\begin{split}
\frac{d\mu_X^{x_0}}{d\mu_{\sigma W}^{x_0}}(x)&=\exp \left\{  - \frac{1}{\sigma}\left( \frac{U(x(T))-U(x_0)}{\sigma}-\frac{1}{2}\int_0^T\sigma\triangle U(x(t))dt \right)-\frac{1}{2}\int_0^T\frac{|\nabla U(x(t))|^2}{\sigma^2}dt \right\}\\
&=\exp\left\{ - \frac{U(x(T))-U(x_0)}{\sigma^2} + \frac{1}{2}\int_0^T \left(\triangle U(x(t))  - \frac{|\nabla U(x(t))|^2}{\sigma^2} \right)dt \right\}.
\end{split}
\end{equation*}
Now we condition on the boundary value $X_T=x_T$, we find by \cite[Lemma 5.3]{Hairer2007} that
\begin{equation*}
\begin{split}
\frac{d\mu_X^{x_0,x_T}}{d\mu_{\sigma W}^{x_0,x_T}}(x)=C_0\exp\left\{ \frac{1}{2}\int_0^T \left(\triangle U(x(t)) - \frac{|\nabla U(x(t))|^2}{\sigma^2} \right)dt \right\},
\end{split}
\end{equation*}
where $C_0$ is a normalized constant, depending only on $x_0, x_T, l$, $\sigma$ and $U$. This result has been used in \cite{Lu2017,Pinski2012}.

For each $\psi\in C^2_{x_0,x_T}[0,T]$ and $x\in \bar{K}_T(\psi,\delta)$, there exists $h\in \{z\in C_0[0,T]\mid \|z\|_T\leq \delta\}$ such that
\begin{equation*}
x=\psi+h,
\end{equation*}
and
\begin{equation*}
\begin{split}
&\ \left| \int_0^T\left(\triangle U  -  \frac{|\nabla U|^2}{\sigma^2}\right)(x(t))dt-\int_0^T\left(\triangle U -  \frac{|\nabla U|^2}{\sigma^2}\right)(\psi(t))dt \right|\\
=&\ \left|\int_0^T\left(\triangle U  - \frac{|\nabla U|^2}{\sigma^2}\right)(\psi(t)+h(t))dt-\int_0^T\left(\triangle U - \frac{|\nabla U|^2}{\sigma^2}\right)(\psi(t))dt\right|\\
\leq&\ C_1 T\delta.
\end{split}
\end{equation*}
where
$$C_1=\sup_{x\in\bar{K}_T(\psi,\delta)} \sup_{t\in[0,T]} \left|\nabla \left( \triangle U  -  \frac{|\nabla U|^2}{\sigma^2} \right)(x(t))\right|.$$
So we know that
\begin{equation}\label{est-1}
\begin{split}
\mu_X^{x_0,x_T}(\bar{K}_T(\psi,\delta))
=&\ \int_{x\in \bar{K}_T(\psi,\delta)}C_0\exp\left\{ \frac{1}{2}\int_0^T\left(\triangle U(x(t))  - \frac{|\nabla U(x(t))|^2}{\sigma^2}\right)dt\right\}d\mu_{\sigma W}^{x_0,x_T}(x)\\
\leq&\ C_0\exp\left\{C_1T\delta + \frac{1}{2}\int_0^T\left(\triangle U(\psi(t))  - \frac{|\nabla U(\psi(t))|^2}{\sigma^2}\right)dt\right\}\mu_{\sigma W}^{x_0,x_T}(\bar{K}_T(\psi,\delta)).
\end{split}
\end{equation}

Let $p_W(\cdot,t|\cdot,s)$ $(0\leq s<t\leq T)$ denote the transition density of Brownian motion $\sigma W$. For each $\psi\in C^2_{x_0,x_T}[0,T]$ and tube size $\delta>0$, according to Lemma \ref{cylinderset} and equation \eqref{finitedistribution} we have that
\begin{equation}\label{est-2}
\begin{split}
&\ p_W(x_T,T|x_0,0)\mu_{\sigma W}^{x_0,x_T}(\bar{K}_T(\psi,\delta))\\
=&\ p_W(x_T,T|x_0,0)\lim_{n\rightarrow\infty}\mu_{\sigma W}^{x_0,x_T}(\bar{I}_n(\psi,\delta))\\
=&\ \lim_{n\rightarrow\infty}\int_{\{z_i\in \bar{B}(\psi(t_i),\delta),i=1,\cdots,n\}}\left(\frac{1}{\sqrt{(2\pi)^k\sigma^2\Delta_i t}}\right)^{n+1}\exp\left\{-\sum_{i=1}^{n+1}\frac{|z_i-z_{i-1}|^2}{2\sigma^2\Delta_i t}\right\}dz_1\cdots dz_{n}\\
&\ (z_0=x_0,z_{n+1}=x_T)\\
=&\ \lim_{n\rightarrow\infty}\int_{\{y_i\in \bar{B}(0,\delta),i=1,\cdots,n\}}\left(\frac{1}{\sqrt{(2\pi)^k\sigma^2\Delta_i t}}\right)^{n+1}\\
&\ \cdot \exp\left\{-\sum_{i=1}^{n+1}\frac{|y_i+\psi(t_i)-y_{i-1}-\psi(t_{i-1})|^2}{2\sigma^2\Delta_i t}\right\}dy_1\cdots dy_{n}\\
&\ (\mbox{Variable substitution}:~y_i=z_i-\psi(t_{i}),~i=0,\cdots,n+1,~\mbox{in particular},~y_0=y_{n+1}=0.)\\
=&\ \lim_{n\rightarrow\infty} \exp\left\{-\sum_{i=1}^{n+1}\frac{|\psi(t_i)-\psi(t_{i-1})|^2}{2\sigma^2\Delta_i t}\right\}\int_{\{y_i\in \bar{B}(0,\delta),i=1,\cdots,n\}}\left(\frac{1}{\sqrt{(2\pi)^k\sigma^2\Delta_i t}}\right)^{n+1}\\
&\ \cdot \exp\left\{-\sum_{i=1}^{n+1}\frac{|y_i-y_{i-1}|^2}{2\sigma^2\Delta_i t}\right\}\exp\left\{-\sum_{i=1}^{n+1}\frac{(y_i-y_{i-1})\cdot(\psi(t_i)-\psi(t_{i-1}))}{\sigma^2\Delta_i t}\right\}dy_1\cdots dy_{n}\\
\leq&\ \exp\left\{\frac{T\delta\|\ddot{\psi}\|_T}{\sigma^2}\right\}\lim_{n\rightarrow\infty}\exp\left\{-\sum_{i=1}^{n+1}\frac{|\psi(t_i)-\psi(t_{i-1})|^2}{2\sigma^2\Delta_i t}\right\}\\
&\ \int_{\{y_i\in \bar{B}(0,\delta),i=1,\cdots,n\}}\left(\frac{1}{\sqrt{2\pi\sigma^2\Delta_i t}}\right)^{n+1}\exp\left\{-\sum_{i=1}^{n+1}\frac{|y_i-y_{i-1}|^2}{2\sigma^2\Delta_i t}\right\}dy_1\cdots dy_{n}\\
=&\ p_W(y_n,l|y_0,0)\exp\left\{\frac{T\delta\|\ddot{\psi}\|_T}{\sigma^2}\right\}\lim_{n\rightarrow\infty} \exp\left\{-\frac{1}{2\sigma^2}\sum_{i=1}^{n+1} \left|\frac{\psi(t_i)-\psi(t_{i-1})}{\Delta_i t}\right|^2\Delta_i t\right\}\mu_{\sigma W}^{y_0,y_n}(\bar{I}_n(0,\delta))\\
=&\ p_W(0,T|0,0)\exp\left\{\frac{T\delta\|\ddot{\psi}\|_T}{\sigma^2}\right\} \exp\left\{-\frac{1}{2}\int_0^T\frac{|\dot{\psi}|^2}{\sigma^2}dt\right\}\mu_{\sigma W}^{0,0}(\bar{K}_T(0,\delta)),
\end{split}
\end{equation}
where $\Delta_i t=t_i-t_{i-1}$, and we have used the discrete version of integration by parts to estimate the cross terms:
\begin{equation*}
\begin{split}
&\ \left|\sum_{i=1}^{n+1}\frac{(y_i-y_{i-1})\cdot(\psi(t_i)-\psi(t_{i-1}))}{\Delta_i t}\right|\\
=&\ \left|y_{n+1}\frac{\psi(t_{n+1})-\psi(t_{n})}{\Delta_i t}-y_0\frac{\psi(t_1)-\psi(t_{0})}{\Delta_i t}+\sum_{i=1}^{n}y_i\frac{\frac{\psi(t_i)-\psi(t_{i-1})}{\Delta_i t}-\frac{\psi(t_{i+1})-\psi(t_i)}{\Delta_i t}}{\Delta_i t}\Delta_i t\right|\\
\leq&\ \sum_{i=1}^{n}|y_i|\left|\frac{\frac{\psi(t_i)-\psi(t_{i-1})}{\Delta_i t}-\frac{\psi(t_{i+1})-\psi(t_i)}{\Delta_i t}}{\Delta_i t}\Delta_i t\right|\\
\leq&\ \delta\|\ddot{\psi}\|_T\sum_{i=1}^{n}|\Delta_i t|\leq T\delta \|\ddot{\psi}\|_T.
\end{split}
\end{equation*}

Now we combine \eqref{est-1} and \eqref{est-2} to derive that
\begin{equation*}
\begin{split}
\mu_X^{x_0,x_T}(\bar{K}_T(\psi,\delta))
\leq&\ C_0\exp\left\{C_1T\delta  + \frac{1}{2}\int_0^T\left(\triangle U(\psi(t)) - \frac{|\nabla U(\psi(t))|^2}{\sigma^2}\right)dt\right\}\\
&\ \cdot \frac{p_W(0,T|0,0)}{p_W(x_T,T|x_0,0)}\exp\left\{\frac{T\delta\|\ddot{\psi}\|_T}{\sigma^2}\right\}\exp\left\{-\frac{1}{2}\int_0^T\frac{|\dot{\psi(t)}|^2}{\sigma^2}\right\}\mu_{\sigma W}^{0,0}(\bar{K}_T(0,\delta))\\
=&\ C_0\exp\left\{C_1T\delta+ \frac{T\delta\|\ddot{\psi}\|_T}{\sigma^2}\right\} \frac{p_W(0,T|0,0)}{p_W(x_T,T|x_0,0)}\\
&\ \cdot \exp\left\{-\frac{1}{2} \int_0^T\left(\frac{|\dot{\psi}(t)|^2}{\sigma^2}  - \triangle U(\psi(t))+\frac{|\nabla U(\psi(t))|^2}{\sigma^2}\right)dt\right\}\mu_{\sigma W}^{0,0}(\bar{K}_T(0,\delta))\\
=&\ C_0\exp\left\{C_1T\delta+\frac{T\delta\|\ddot{\psi}\|_T}{\sigma^2}  + \int_0^T \frac{\dot{\psi}(t) \cdot \nabla U(\psi(t))}{\sigma^2} dt \right\} \frac{p_W(0,T|0,0)}{p_W(x_T,T|x_0,0)}\\
 &\ \cdot \exp\{-S_X^{OM}(\psi)\}\mu_{\sigma W}^{0,0}(\bar{K}_T(0,\delta)) \\
=&\ C_0\exp\left\{C_1T\delta+\frac{T\delta\|\ddot{\psi}\|_T}{\sigma^2}  + \frac{U(x_T)-U(x_0)}{\sigma^2}\right\}\frac{p_W(0,T|0,0)}{p_W(x_T,T|x_0,0)}\\
&\ \cdot \exp\{-S_X^{OM}(\psi)\}\mu_{\sigma W}^{0,0}(\bar{K}_T(0,\delta)).
\end{split}
\end{equation*}
Similarly, we have that
\begin{equation*}
\begin{split}
\mu_X^{x_0,x_T}(\bar{K}_T(\psi,\delta))
\geq&\ C_0 \exp\left\{-C_1T\delta-\frac{T\delta\|\ddot{\psi}\|_T}{\sigma^2}  + \frac{U(x_T)-U(x_0)}{\sigma^2}\right\} \frac{p_W(0,T|0,0)}{p_W(x_T,T|x_0,0)} \\
&\ \cdot \exp\{-S_X^{OM}(\psi)\} \mu_{\sigma W}^{0,0}(\bar{K}_T(0,\delta)).
\end{split}
\end{equation*}
These give the desired results of the Lemma \ref{OMderived} with
\begin{equation*}
  C = C_0 \exp\left\{ \frac{U(x_T)-U(x_0)}{\sigma^2}\right\} \frac{p_W(0,T|0,0)}{p_W(x_T,T|x_0,0)}.
\end{equation*}

\begin{remark}
(i). In \cite{Durr1978,Ikeda1980} the OM functional was derived from the measure $\mu_X^{x_0}$ by using Girsanov formula twice. This works since $\mu_X^{x_0}$ is absolutely continuous with respect to $\mu_{\sigma W}^{x_0}$, and both measures are quasi-translation invariant (see \cite{Durr1978} for details). However, the bridge measures $\mu_X^{x_0,x_T}$ and $\mu_{\sigma W}^{x_0,x_T}$ are not quasi-translation invariant in $C_{x_0}[0,T]$. This is the difference between our method and the methods in \cite{Durr1978,Ikeda1980} to derive the OM functional.

(ii). Note that when we applied Girsanov transform in \eqref{Girsanov}, the trajectory $x(\cdot)$ at the right hand side (RHS) is understood as a trajectory of Brownian motion under $\mu_{\sigma W}^{x_0}$, instead of that of the solution process $X$.

(iii).
The authors of \cite{MP16} used a one-dimensional potential $U$ that has two wells, one narrow and one broad, to investigate trajectory sampling schemes. They observed that the Laplacian term $U''$ has a huge influence on the OM functional. In an ensemble that was generated from a hybrid Monte Carlo sampling scheme, the resulting trajectories spent the largest amount of time in the narrow well, where the curvature of $U$ is the largest. They pointed out that such a result is \emph{unphysical}, by the reason that such trajectories are inconsistent with the equilibrium thermodynamical distribution where the particle would be expected to be found in the broad well in most of the time. Other potentials that appeared to be unphysical may also be found in \cite{PS10}. Possible interpretations for the origin of unphysical phenomena of MPPs/MPTPs are the following: firstly, minimizing an OM functional is in general a non-convex optimization problem as its Lagrangian is non-convex (at least for double-well potentials), and the Euler--Lagrange equation for an OM functional, which may have many solutions, is merely a necessary but not sufficient condition for its minimizers, so numerical schemes based on this approach can only find (an approximation of) some (local) minimizers which we are not able to know if is global, except for some special cases like the linear one, cf. Section \ref{application}; secondly, the first author of the present paper showed with his collaborators in \cite{HCW21} that it can occur that the (local) minimizers for the OM functional with a double-well potential have multiple back-and-forth transitions when the given transition time $T$ is large; thirdly, it was proved in \cite{DLL21,HV08} that when transition time $T$ goes to infinity (which is the case for \cite{MP16,PS10} since they considered equilibrium states), the MPTP converges under the \emph{Fr\'echet distance} which allows nonhomogeneous time parametrizations on the path space without affecting the values of OM functional, thus the time duration in wells is unimportant in this context and cannot be a good reason for the unphysicalness.
\end{remark}

\subsection{A sufficient and necessary characterization for most probable transition paths}\label{MPTPequivalent}


Note that the drift term in \eqref{newsde} is singular at time $t=T$. In fact, it is this singular attractive potential that forces all the trajectories of $Y$ to $x_T$ at time $T$ \cite{Cetin2016}. In other words, the process $Y$ must ``transit'' to $x_T$ at time $T$. So formally we do not need to emphasize the transition behavior for the process $Y$. That is, the problem reduces to: among all possible smooth paths starting at $x_0$, which one is most probable for the solution process $Y$ of $(\ref{newsde})$?

The solution process $Y$ of $(\ref{newsde})$ induces a measure $\mu_Y^{x_0}$ on $\mathcal{B}^{x_0}_{[0,T]}$ as its law by
\begin{equation*}
\begin{split}
\mu_Y^{x_0}(B)&=\mathbb{\tilde P}^{x_0}\left( \{w\in\tilde{\Omega}\mid Y (\omega)\in B\} \right),\quad B\in\mathcal{B}^{x_0}_{[0,T]}.
\end{split}
\end{equation*}

\begin{remark}
Since we have known that $\mathbb{\tilde P}^{x_0}(Y_T=x_T)=1$, the most probable paths of system $(\ref{newsde})$ must reach point $x_T$ at time $T$.
\end{remark}

To figure out the relation between the most probable transition paths of $X$ and the most probable paths of $Y$, we need the help of the bridge measure $\mu_{X}^{x_0,x_T}$. Note that although the two systems $(\ref{newsde})$ and \eqref{sdetwobridge} may be defined on different probability spaces, their associate induced measures $\mu_Y^{x_0}$ and $\mu_X^{x_0,x_T}$ are defined on the same path space $(C_{x_0}[0,T],\mathcal{B}^{x_0}_{[0,T]})$. The following lemma gives the relation between measures $\mu_Y^{x_0}$ and $\mu_{X}^{x_0,x_T}$.

\begin{lemma}[Bridge measures $=$ laws of bridge SDEs]\label{equivalence}
The two measures $\mu_{X}^{x_0,x_T}$ and $\mu_Y^{x_0}$ coincide.
\end{lemma}
\begin{proof}
The equations $(\ref{conditionalp})$ and $(\ref{densityY})$ show us that the transition density functions of process $X$ under $\mathbb{P}^{x_0}(\cdot\mid X_T=x_T)$ and process $Y$ under $\mathbb{\tilde P}^{x_0}$ are identical. Let $I=\{\psi\in C_{x_0}[0,T]\mid \psi(t_1)\in E_1,\cdots,\psi(t_n)\in E_n\}$ be a cylinder set with $0\leq t_1<t_2<\cdots<t_n\le T$ and Borel sets $E_i \subset \mathbb{R}^k$. In the case that $t_n<T$, we have the following equalities:
\begin{equation*}
\begin{split}
\mu_Y^{x_0}(I)
=&\ \mathbb{\tilde P}^{x_0}(Y_{t_i}\in E_i,i=1,\cdots,n)\\
=&\ \int_{E_1}\cdots\int_{E_n}p^{x_0,x_T}(y_1,t_1|x_0,0)\cdots p^{x_0,x_T}(y_n,t_n|y_{n-1},t_{n-1})dy_1\cdots dy_n\\
=&\ \mathbb{P}^{x_0}(X_{t_i}\in E_i,~i=1,\cdots,n|~X_T=x_T)\\
=&\ \mu_{X}^{x_0,x_T}(I).
\end{split}
\end{equation*}
In the case that $t_n=T$, due to the fact $\mathbb{\tilde P}^{x_0}(Y_T=x_T)=1$ and $\mathbb{P}^{x_0}(X_T=x_T|X_T=x_T)=1$, we know that
\begin{equation*}
\begin{aligned}
\mu_Y^{x_0}(I)&=\mu_Y^{x_0}(\{\psi\in C_{x_0}[0,T]\mid\psi(t_i)\in E_i,i=1,\cdots,n-1\})\\
&=\mu_{X}^{x_0,x_T}(\{\psi\in C_{x_0}[0,T]\mid\psi(t_i)\in E_i,i=1,\cdots,n-1\})=\mu_{X}^{x_0,x_T}(I),\quad\mbox{if }x_T\in E_n,\\
\mu_Y^{x_0}(I)&=0=\mu_{X}^{x_0,x_T}(I),\quad\mbox{if }x_T\notin E_n.
\end{aligned}
\end{equation*}
Thus the measures $\mu_Y^{x_0}$ and $\mu_{X}^{x_0,x_T}$ coincide on all cylinder sets of $C_{x_0}[0,T]$. Recall that, the field $\mathcal{B}^{x_0}_{[0,T]}$ is the $\sigma$-field generated by all cylinder sets. By the Carath\'{e}odory measure extension theorem, we know that the two probability measures $\mu_Y^{x_0}$ and $\mu_{X}^{x_0,x_T}$ coincide on $\mathcal{B}^{x_0}_{[0,T]}$. This completes the proof.
\end{proof}

\begin{remark}\label{remark-2}

The bridge measure $\mu_{X}^{x_0,x_T}$ contains the thermodynamic (or statistical) information of all trajectories of the underlying process $X$ that happens to end at the particular metastable state $x_T$, while the law $\mu_Y^{x_0}$ of bridge SDE \eqref{newsde} contains the information of the process $Y$ that are forced to end at the state $x_T$ by an additional virtual potential, cf. Remark \ref{remark-1}. Lemma \ref{equivalence} tells that the thermodynamic information of the above two different thermodynamic systems is equivalent.
\end{remark}

Under Lemmas \ref{OMderived} and \ref{equivalence}, for $\psi_1,\psi_2\in C^2_{x_0,x_T}[0,T]$ we have that
\begin{equation*}
\begin{split}
\lim_{\delta\downarrow0} \frac{\mu_Y^{x_0}(\bar K_T(\psi_1,\delta))}{\mu_Y^{x_0}(\bar K_T(\psi_2,\delta))} =&\ \lim_{\delta\downarrow0} \frac{\mu_X^{x_0,x_T}(\bar K_T(\psi_1,\delta))}{\mu_X^{x_0,x_T}(\bar K_T(\psi_2,\delta))}
 = \exp(S^{OM}_X(\psi_2)-S^{OM}_X(\psi_1)) =\lim_{\delta\downarrow0} \frac{\mu_X^{x_0}(\bar{K}_T(\psi_1,\delta))}{\mu_X^{x_0}(\bar{K}_T(\psi_2,\delta))}.
\end{split}
\end{equation*}
Thus we know that if we want to find the most probable transition path of system \eqref{SDE}, an alternative way is to find the most probable path of system \eqref{newsde}. Moreover, thanks to Lemma \ref{MPP}, the most probable path of \eqref{newsde} can be characterized by its corresponding deterministic ODE. We summarize these conclusions in the following theorem, as the main result of this paper. 

\begin{theorem}[Sufficient \& necessary characterization for MPTPs]\label{maintheory}
Suppose that the assumptions \ref{wellposedfirst} and \ref{assumptiononEU} hold. \\
(i). The most probable transition path(s) of the system \eqref{SDE} coincide(s) with the most probable transition path(s) of the associated bridge SDE (\ref{newsde}). \\
(ii). A path $\psi^*\in C^2_{x_0,x_T}[0,T]$ is a most probable transition path of \eqref{SDE} if and only if it solves the following first-order ODE
\begin{equation}\label{MPTP}
   d\psi^*(t)=\left[- \nabla U(\psi^*(t))+\sigma^2\nabla\ln p(x_T,T|\psi^*(t),t) \right]dt, \quad t\in(0,T),\quad \psi^*(0)=x_0,
\end{equation}
where $p(\cdot,\cdot|\cdot,\cdot)$ is the transition density of the solution process of \eqref{SDE}.
\end{theorem}

\begin{remark}
(i). By the same reason as Remark \ref{rmk-a}-(i), under the assumption of the above theorem, the existence and uniqueness of a most probable transition path of \eqref{SDE} are equivalent to the existence and uniqueness of a solution of \eqref{MPTP} in $C^2_{x_0}[0,T]$. In a way, we transform a second-order ODE, the Euler--Lagrange equation associated with the OM functional, to a first-order one \eqref{MPTP}.

(ii). As we have mentioned in the introduction, neither the most probable transition paths of system \eqref{SDE} nor the most probable paths of the bridge SDE \eqref{newsde} are genuine trajectories since they are differentiable curves. This means the MPPs/MPTPs are statistical predictions rather than actual observations or measurements. Theorem \ref{maintheory} infers that such statistical predictions for the two thermodynamic systems in Remark \ref{remark-2} are equivalent. In many cases, the paths matter, however, changes in the thermodynamic properties depend only on the initial and final states and not upon the paths \cite{BS20}.
\end{remark}

Since the first-order ODE \eqref{MPTP} is sufficient and necessary for MPTPs, it must fulfill the following Euler--Lagrange equation of the OM functional \eqref{OM},
\begin{equation}\label{EL}
\left\{ \begin{aligned}
  & \ddot{\psi} +\frac{1}{2} \nabla \left(\sigma^2\Delta U(\psi^*)-|\nabla U(\psi)|^2\right)=0,\\
  & \psi(0)=x_0,\quad \psi(T)=x_T.
\end{aligned}\right.
\end{equation}
We differentiate both sides of \eqref{MPTP} as follows,
\begin{equation*} 
\begin{aligned}
   \ddot{\psi}^*(t) =&\  - \nabla^2 U(\psi^*(t))  \cdot \dot{\psi}^*(t) + \sigma^2\frac{d }{dt}\nabla\ln  p(x_T,T|\psi^*(t),t) \\
   =&\  \nabla^2 U(\psi^*(t)) \cdot \left[ \nabla U(\psi^*(t)) - \sigma^2\nabla\ln  p(x_T,T|\psi^*(t),t) \right] +\sigma^2\frac{d }{dt}\nabla\ln  p(x_T,T|\psi^*(t),t).
\end{aligned}
\end{equation*}
Comparing the above equation with \eqref{EL} and denoting
\begin{equation}\label{beta}
  \beta(t) := \sigma^2 \nabla\ln  p(x_T,T|\psi^*(t),t), \quad t\in[0,T),
\end{equation}
we immediately have
\begin{corollary}
  Under the same assumptions as Theorem \ref{maintheory}, suppose that $\psi^*\in C^2_{x_0,x_T}[0,T]$ is a most probable transition path of \eqref{SDE}. Let $p(\cdot,\cdot|\cdot,\cdot)$ be the transition density of \eqref{SDE}. Then the curve $\beta: [0,T)\to \R^k$ defined in \eqref{beta} satisfies the following ODE:
\begin{equation}\label{ODE-momentum}
    \dot \beta -\nabla^2 U(\psi^*) \cdot \beta + \frac{\sigma^2}{2} \nabla \Delta U(\psi^*) = 0, \quad \beta(0) = \sigma^2 \nabla\ln  p(x_T,T|x_0,0).
\end{equation}
\end{corollary}

The ODE \eqref{ODE-momentum}, or equivalently, the fact that solutions of the first-order ODE \eqref{MPTP} also solve the EL equation \eqref{EL}, cannot be apparently verified in a straightforward way in general, since the information of the transition density $p$ is obscure. Our results amount to an indirect probabilistic way to prove the implication of the first-order ODE \eqref{MPTP} to the EL equation \eqref{EL}. We will verify this implication for some special cases either analytically or numerically in the next two sections.


\section{Exact and approximate equations for most probable transition paths}\label{application}

At this stage, to characterize the most probable transition path of system \eqref{SDE}, what we need to do is to evaluate the transition density $p$. But it is impossible to get an analytical expression for this density for general nonlinear stochastic systems. However, in some special cases, we can obtain the analytical expression or an approximation for the transition density.

In this section, we will consider three classes of stochastic systems and study their transition densities: linear systems, Hongler's model \cite{hongler1981study} and general nonlinear systems.

\subsection{Linear systems}
Consider the following linear equation \cite[Section 5.6]{Karatzas1991}:
\begin{equation}\label{linearequation}\left\{
\begin{aligned}
dX_t&=[GX_t+a]dt+\sigma dW_t,\quad 0\leq t<\infty\\
X_0&=x_0,
\end{aligned}\right.
\end{equation}
where $W$ is a $k$-dimensional Brownian motion independent of the initial vector $x_0\in\mathbb R^k$, $G$ is a $k\times k$ constant nondegenerate symmetric matrix, $a$ is a $(k\times 1)$ matrix and the noise intensity $\sigma$ is a positive constant. Under these settings, it is easy to check that the drift term is the gradient of the potential function
\begin{equation}\label{potential}
  U(x) = -\textstyle{\frac{1}{2}} x^TGx - a^T x + \mathrm{constant},
\end{equation}
Since an affine transformation of the action functional does not change its minimizers, we have
\begin{proposition}
  The most probable transition path of the linear system \eqref{linearequation} connecting two given states $x_0$ and $x_T$, is a path $\psi^*$ that makes the following  Freidlin--Wentzell (FW) action functional achieve its minimum value in the path space $C^2_{x_0,x_T}[0,T]$,
\begin{equation}\label{FW}
\begin{split}
S_{X}^{FW}(\psi)=\frac{1}{2}\int_0^T|\dot{\psi}(s) + \nabla U(\psi(s))|^2ds.
\end{split}
\end{equation}
Consequently, $\psi^*$ satisfies the following Euler--Lagrange equation associated to \eqref{FW},
\begin{equation}\label{EL-linear}
    \ddot \psi = G^2 \psi + Ga, \quad \psi(0) = x_0,\ \psi(T) = x_T.
\end{equation}
\end{proposition}

The solution of the system \eqref{linearequation} has the following representation,
\begin{equation}\label{sol-linear-sde}
  \begin{split}
    X_t 
    &= e^{tG} x_0 + G^{-1} \left( e^{tG}a - a \right) + \sigma\int_0^t e^{(t-s)G} dW_s, \quad 0\leq t<\infty.
  \end{split}
\end{equation}
where $\{e^{tG}\}_{t\ge0}$ is the matrix semigroup of $G$, i.e., $\Phi(t) = e^{tG}$ is the solution of the differential equation
\begin{equation*}
\left\{
\begin{aligned}
\dot{\Phi}(t)&=G\Phi(t)\\
\Phi(0)&=I,
\end{aligned}\right.
\end{equation*}
where $I$ is the $k\times k$ identity matrix.
Clearly, the solution $X$ in \eqref{sol-linear-sde} is a Gaussian process, whose mean vector and covariance matrix are given by
\begin{gather*}
\mu(t;x_0) \triangleq \mathbb{E}X_t= 
e^{tG} x_0 + G^{-1} \left( e^{tG}a - a \right),\\
\Sigma(t)\triangleq \mathbb{E}\left[(X_t-\mathbb{E}X_t)(X_t-\mathbb{E}X_t)^T\right] = 
\frac{\sigma^2}{2} G^{-1} \left( e^{2tG} - I \right),
\end{gather*}
Hence, the probability density function of $X$ is
\begin{equation*}
\begin{split}
 p(x,t|x_0,0)=&\ \frac{1}{(2\pi)^{k/2} \sqrt{\det \Sigma(t) }} \exp\Bigg\{ -\frac{1}{2}\left[x-\mu(t;x_0)\right]^T \Sigma(t)^{-1} \left[x-\mu(t;x_0)\right]\Bigg\},
\end{split}
\end{equation*}
and the associated curve $\beta$ in \eqref{beta} is
\begin{equation}\label{beta-linear}
  \begin{split}
    \beta(t) 
    &= 2 G \left( e^{2(T-t)G} - I \right)^{-1} e^{(T-t) G} \left[ x_T- e^{(T-t) G} \psi^*(t)- G^{-1} \left( e^{(T-t) G} a- a \right) \right].
  \end{split}
\end{equation}



\begin{corollary}\label{linearMPTP}
The most probable transition path of the linear system \eqref{linearequation} is described by the following ordinary differential equation
\begin{equation}\label{linear-first-ode}
\left\{
  \begin{aligned}
   \dot{\psi}^*(t)
   =&\ G\psi^*(t)+a + 2 G \left( e^{2(T-t)G} - I \right)^{-1} e^{(T-t) G} \left[ x_T- e^{(T-t) G} \psi^*(t)- G^{-1} \left( e^{(T-t) G} a- a \right) \right], \ t\in[0,T), \\
   \psi^*(0)=&\ x_0.
  \end{aligned}\right.
\end{equation}
\end{corollary}

\begin{remark}
  Observe that for the linear case, both the EL equation \eqref{EL-linear} and the first-order ODE \eqref{linear-first-ode} are independent of the noise intensity constant $\sigma$. This may not be true for general nonlinear cases.
\end{remark}
Since, by matrix calculus,
\begin{equation*}
  \begin{split}
    \frac{d}{dt} \left( e^{2(T-t)G} - I \right)^{-1} &= -\left( e^{2(T-t)G} - I \right)^{-1} \left[ \frac{d}{dt} \left( e^{2(T-t)G} - I \right) \right] \left( e^{2(T-t)G} - I \right)^{-1} \\
    &= 2 \left( e^{2(T-t)G} - I \right)^{-1} G e^{2(T-t)G} \left( e^{2(T-t)G} - I \right)^{-1},
  \end{split}
\end{equation*}
we differentiate the function $\beta$ in \eqref{beta-linear} and get
\begin{equation*}
  \begin{split}
    \dot\beta(t) =&\ - G \beta(t) + 2 \left( e^{2(T-t)G} - I \right)^{-1} G e^{2(T-t)G} \beta(t) \\
    &\ + 2 G \left( e^{2(T-t)G} - I \right)^{-1} e^{(T-t) G} \left[ G e^{(T-t) G} \psi^*(t) - e^{(T-t) G} ( G\psi^*(t)+a + \beta(t) ) + e^{(T-t) G} a \right] \\
    =&\ -G \beta(t),
  \end{split}
\end{equation*}
which verifies \eqref{ODE-momentum}.
We then differentiate both sides of \eqref{linear-first-ode} as follows,
\begin{equation*} 
  \ddot{\psi}^*(t) = G\dot\psi^*(t) + \dot\beta(t) = G[ G\psi^*(t)+a + \beta(t) ] - G \beta(t) = G^2 \psi^*(t)+ Ga,
\end{equation*}
that is, \eqref{linear-first-ode} implies the EL equation \eqref{EL-linear}.

\subsection{Hongler's model}\label{sec-4-2}
Hongler's model \cite{hongler1981study} is described by the following 1-dimensional nonlinear SDE:
\begin{equation}\label{hongler}
    \begin{aligned}
        dX_t=-\frac{dU}{dx}(X_t)dt+dW_t,\quad X_0=x_0,
    \end{aligned}
\end{equation}
where
\begin{equation}\label{honglerU1st}
    \begin{aligned}
        U(x)= \frac{\sqrt{A}}{2} x^2-\ln \left[ {}_1F_1\left(\frac{B}{4\sqrt{A}}+\frac{1}{4};\frac{1}{2};\sqrt{A}x^2 \right) \right],
    \end{aligned}
\end{equation}
and ${}_1F_1$ stands for the confluent hypergeometric function (see Appendix \ref{appendixhongler}) and the constants $A$ and $B$ are such that 
\begin{equation*} 
    \begin{aligned}
        A\geq0,\quad B\geq-\sqrt{A}.
    \end{aligned}
\end{equation*}
The transition density function of the solution process of \eqref{hongler} is obtained therein as
\begin{equation}\label{densityhongler}
\begin{split}
p(x,t|y,s)=&\ \exp\left \{ \frac{\sqrt{A}}{e^{-2\sqrt{A}(t-s)}-1} \left(x-ye^{-\sqrt{A}(t-s)}\right)^2 -\frac{B}{2}(t-s) \right\} {}_1F_1\left(\frac{B}{4\sqrt{A}}+\frac{1}{4};\frac{1}{2};\sqrt{A}x^2 \right)\\
&\ \times \left[ \sqrt{ \frac{2\pi}{\sqrt{A}} \sinh\left(\sqrt{A}(t-s)\right)        }    {}_1F_1\left(\frac{B}{4\sqrt{A}}+\frac{1}{4};\frac{1}{2};\sqrt{A}y^2 \right)\right]^{-1}.
\end{split}
\end{equation}
Hongler's model can be generalized to higher dimensional cases, see e.g. \cite{garrido1985exact}.

Substituting the above density \eqref{densityhongler} to \eqref{MPTP} we obtain that the ODE for the MPTPs of Hongler's SDE \eqref{hongler} is given by
\begin{equation}\label{honglerMPTP}
\begin{aligned}
    \dot \psi^*(t)=&\ -\sqrt{A}\psi^*(t) +\frac{2\sqrt{A} e^{-\sqrt{A}(T-t)}}{1-e^{-2\sqrt{A}(T-t)}  } \left( x_T-\psi^*(t)e^{-\sqrt{A}(T-t)} \right),\quad t\in(0,T),\quad \psi^*(0)=x_0.
    \end{aligned}
\end{equation}
See Appendix \ref{appendixhongler} for detailed derivations of \eqref{honglerMPTP}.
An interesting observation here is that equation \eqref{honglerMPTP} has exactly the same form as the first-order ODE \eqref{linear-first-ode} of MPTPs for the linear case, with $G=-\sqrt A$ and $a=0$, although the system \eqref{hongler} is nonlinear and its transition density \eqref{densityhongler} is much more complicated than the linear case.

On the other hand, it has been observed in \cite[Section 2]{hongler1981study} that, the Euler--Lagrange equation for the OM functional of \eqref{hongler} also has the same form as that of the linear case \eqref{EL-linear}, with $G=-\sqrt A$ and $a=0$, i.e.,
\begin{equation}\label{hongler-EL}
\ddot \psi=A\psi, \quad \psi(0) = x_0,\ \psi(T) = x_T,
\end{equation}
of which the general solution is given by
\begin{equation*}
\psi(t)=C_1 e^{-\sqrt{A}t}+C_2e^{\sqrt{A}t},
\end{equation*}
where $C_1,~C_2$ are two constants determined by boundary conditions. Therefore, the implication of the first-order ODE \eqref{honglerMPTP} to the EL equation \eqref{hongler-EL}, for Hongler's nonlinear model, follows in the same way as the linear case. We also remark from the EL equation (or from \eqref{honglerMPTP}) that the MPTPs are independent of the choice of the constant $B$.

In \cite{hongler1982exact}, Hongler studied the minus of the potential \eqref{honglerU1st}, that is,
\begin{equation*} 
    \begin{aligned}
        U(x)=2\ln \left\{e^{-\frac{1}{2}x^2}  {}_1F_1\left(\frac{A}{2}+\frac{1}{4},\frac{1}{2},\frac{1}{2}x^2 \right) \right\},
    \end{aligned}
\end{equation*}
with condition $A\geq-\frac{1}{2}$. For such potential, the transition density is also exact but in a more involved form. 

In \cite{romero1990diffusion}, the following potential was considered:
\begin{equation*} 
    \begin{aligned}
        U(x)=-2\ln \left\{\left(\frac{1}{2}x^2\right)^{a+\frac{1}{4}}e^{-\frac{1}{4}x^2}  {}_1F_1\left(A,B,\frac{1}{2}x^2 \right) \right\},
    \end{aligned}
\end{equation*}
where
\begin{equation*} 
    \begin{aligned}
        A=\frac{1}{2}+a-\frac{1}{8}b,\quad B=1+2a,
    \end{aligned}
\end{equation*}
with
\begin{equation*} 
    \begin{aligned}
        a>-\frac{1}{2},\quad A\geq0.
    \end{aligned}
\end{equation*}
This can be regarded as a generalization of Hongler's original model since when $a=-\frac{1}{4}$ the potential is the same as \eqref{honglerU1st}. For the case in which $a\neq-\frac{1}{4}$, the stochastic system is defined on either $\mathbb R_+$ or $\mathbb R_-$ depending on its initial position. The associated transition density function was derived in \cite{romero1990diffusion} for various parameters.

\subsection{Approximations for nonlinear systems}

Now we turn back to the general nonlinear system \eqref{SDE}. 

\subsubsection{Short-time approximation}

The authors in \cite{Orland2011} transformed the Fokker-Planck equation \eqref{forward} for $p$ to the following imaginary time Schr\"odinger equation by the transformation $\Psi(x,t) = \exp(\frac{U(x)}{\sigma^2})p(x,t|x_0,t_0)$,
$$\frac{\partial\Psi}{\partial t}= \frac{\sigma^2}{2}\triangle\Psi(x,t)- \frac{1}{2\sigma^2}V(x)\Psi(x,t),$$
where $V(x)=|\nabla U(x)|^2 - \sigma^2 \triangle U(x)$; they then used Trotter approximation to approximate the Schr\"odinger kernel and obtained the following asymptotic formula,
$$p(x,t|y,s)=\exp\left\{-\frac{U(x)-U(y)}{\sigma^2} -\frac{|x-y|^2}{2\sigma^2 (t-s)}-\frac{t-s}{4\sigma^2}(V(x)+V(y)) \right\} + O(|t-s|^3).$$
This formula not only suggests that the SDE $(\ref{newLDPequation})$ can be approximated by the following SDE
\begin{equation*}
\left\{
\begin{aligned}
d\hat Y_t&=\left[\frac{x_T-\hat Y_t}{T-t}-\frac{T-t}{4}\nabla V(\hat Y_t)\right]dt+\sigma d\tilde{W}_t,\quad t\in(0,T),\\
\hat Y_0&=x_0,
\end{aligned}\right.
\end{equation*}
when the transition time $T$ is short, but also quantifies that the approximation error for $p$ is at most $O(T^3)$. 

In view of these, we have an approximation scheme for the most probable paths of the bridge SDE \eqref{newsde}, or equivalently, for ODE \eqref{MPTP}, as the following first-order differential equation:
\begin{equation*}
\left\{
\begin{aligned}
\frac{d\psi_{\text{appr},1}}{dt}&=\frac{x_T-\psi_{\text{appr},1}}{T-t}-\frac{T-t}{4}\nabla \left( |\nabla U|^2 - \sigma^2\triangle U \right) (\psi_{\text{appr},1}),\quad t\in(0,T),\\
\psi_{\text{appr},1}(0)&=x_0,
\end{aligned}\right.
\end{equation*}

More approximations for transition density functions for general SDEs with additive noise, with more complicated forms, can be found in \cite{ait1999transition,choi2015explicit}. It is worth remarking that the authors of \cite{HCW21} claimed that considering the (most probable) transition time for stochastic dynamical systems with constant noise in finite time intervals makes more sense than that the terminal time $T$ tends to infinity, since the constant noise intensity makes the ``particle'' under study fluctuate away the steady states easier and thus transitions between steady states occur more frequently. When $T<1$, the approximations for probability density $p$, given in \cite{ait1999transition,choi2015explicit}, enables us to approximate the MPTP of \eqref{SDE} or the ODE \eqref{MPTP} with sufficiently small error.

\subsubsection{Small-noise approximation}\label{sec-4-3-2}

We consider the small-noise version of system \eqref{SDE} as follows:
\begin{equation}\label{LDPequation}\left\{
  \begin{aligned}
    dX^{\varepsilon}_t &= - \nabla U(X^{\varepsilon}_t)dt+\varepsilon dW_t, \quad t\in(0,T], \\
    X_0^\varepsilon &=x_0,
  \end{aligned}\right.
\end{equation}
where $\varepsilon>0$ is a small parameter. 
If we let $\varepsilon$ tend to 0, such a system is considered in the framework of large deviation which has been studied with a rich history. See for example \cite{Freidlin2012} and references therein.

The  Freidlin--Wentzell theory of large deviations asserts that, for $\delta$ and $\varepsilon$ positive and sufficiently small,
\begin{equation*}
\begin{split}
\mathbb{P}^{x_0}(\|X^\varepsilon - \psi \|_T<\delta) \sim \exp(-\varepsilon^{-2}S_X^{FW}(\psi)),
\end{split}
\end{equation*}
where $S_X^{FW}$ is the FW action functional in \eqref{FW}. Intuitively $S_X^{FW}$ is the dominant term of OM functional (\ref{OM}). Indeed, we have the following proposition which is a straightforward consequence of $\Gamma$-convergence. 
\begin{proposition}
  As $\varepsilon \downarrow 0$, the most probable transition path of system $(\ref{LDPequation})$ converges to the minimizer of the following functional:
\begin{equation*}
\sup_{N_\psi\in \mathcal N(\psi)} \inf_{\phi\in N_\psi} S_{X}^{FW}(\phi) = \Gamma\text{-}\lim_{\varepsilon \downarrow 0} \varepsilon^2 S_{X}^{OM}(\phi),
\end{equation*}
where $\mathcal N(\psi)$ denotes the set of all neighbourhoods of $\psi$ in $C^2_{x_0,x_T}[0,T]$.
\end{proposition}

The Lagrangian of the  Freidlin--Wentzel action functional \eqref{FW} is 
\begin{equation}\label{Lagrangian-FW}
  L^{FW}(\psi,\dot{\psi}) = \frac{1}{2} |\dot{\psi}  + \nabla U(\psi)|^2
\end{equation}
and the associated Euler--Lagrange equation is a second-order boundary value problem that reads
\begin{equation*}
\begin{cases}
\ddot{\psi}-\frac{1}{2}\nabla |\nabla U(\psi)|^2=0,\\
\psi(0)=x_0,~\psi(T)=x_T.
\end{cases}
\end{equation*}
The classical variational method tells that the Euler--Lagrange equation is a necessary but not sufficient condition of the most probable transition paths.

On the other hand, the bridge process of system $(\ref{LDPequation})$ is
\begin{equation}\label{newLDPequation}\left\{
\begin{aligned}
dY^\varepsilon_t&=\left[ -\nabla U(Y^{\varepsilon}_t)+  \varepsilon^2 \nabla\ln p_\varepsilon(x_T,T|Y^\varepsilon_t,t) \right]dt +  \varepsilon d\tilde{W}_t,\quad t\in(0,T),\\
Y^\varepsilon_0&=x_0,
\end{aligned}\right.
\end{equation}
where $p_\varepsilon(x_T,T|x,t)$ is the transition density of the solution process of system $(\ref{LDPequation})$. The problem here is how to characterize the limit of the term $\varepsilon^2 \nabla\ln p_\varepsilon(x_T,T|Y^\varepsilon_t,t)$ as $\varepsilon\downarrow0$. 

The paper \cite{Delarue2017} for trajectory sampling gave a scheme to approximate the bridge SDE \eqref{newLDPequation} in one-dimensional case, in the limit $\varepsilon\downarrow0$.
The approximating SDE is given by
\begin{equation*}\left\{
\begin{aligned}
d\hat Y^\varepsilon_t&=\left[\frac{x_T-\hat Y^\varepsilon_t}{T-t}-\frac{T-t}{2}\int_{0}^{1}(1-u) \textstyle{\frac{d}{d x}\left[\left(\frac{d U}{d x}\right)^2\right]} \left(x_Tu+\hat Y^\varepsilon_t(1-u) \right)du\right]dt + \varepsilon d\tilde{W}_t,\  t\in(0,T),\\
\hat Y^\varepsilon_0&=x_0.
\end{aligned}\right.
\end{equation*}
As a consequence, we have an approximation scheme for the most probable paths of the system \eqref{newLDPequation}, or for ODE corresponding to \eqref{newLDPequation} by vanishing noise, described by the following integro-differential equation:
\begin{equation*}
\left\{
\begin{aligned}
\frac{d\psi_{\text{appr},2}}{dt}&=\frac{x_T-\psi_{\text{appr},2}}{T-t}-\frac{T-t}{2}\int_{0}^{1}(1-u) \textstyle{\frac{d}{d x}\left[\left(\frac{d U}{d x}\right)^2\right]} \left( x_Tu+\psi_{\text{appr},2}(1-u)\right) du,\ t\in(0,T),\\
\psi_{\text{appr},2}(0)&=x_0.
\end{aligned}\right.
\end{equation*}

\section{Examples and numerical experiments}\label{examples}

Let us consider several examples in order to illustrate our results.

\begin{example}[The free Brownian motion]\label{brownian}

The simplest case is the free particles in Euclidean space. In this case, the Green's function $p(x_T,T|x,t)$ can be written explicitly as follows
\begin{equation*}
p(x_T,T|x,t)=\frac{1}{\sqrt{2\pi(T-t)}}e^{-\frac{(x_T-x)^2}{2(T-t)}}.
\end{equation*}
The corresponding Markovian bridge process is described by the following SDE:
\begin{equation}\label{exp-BM}
dY_t=\frac{x_T-Y_t}{T-t}dt+d\tilde{W}_t,\quad Y_0=x_0.
\end{equation}
The partial derivative of the drift term with respect to the position variable is independent of the position variable. Thus by \eqref{maintheory}, the most probable path of \eqref{exp-BM} is
\begin{equation*}
\frac{d\psi^*}{dt}=\frac{x_T-\psi^*}{T-t},~\psi^*(0)=x_0 ~\Longrightarrow~\psi^*(t)=x_T+\frac{x_0-x_T}{T}(T-t),~t\in[0,T),
\end{equation*}
which can be verified as the extremal path of the OM functional $S^{OM}(\psi)=\frac{1}{2}\int_0^T\dot{\psi}^2ds$ over the path space $C^2_{x_0,x_T}[0,T]$.

\end{example}

\begin{example}[Linear systems]

In this example, we consider the linear system \eqref{linearequation} in one-dimensional and two-dimensional cases.

Case 1. Consider the scalar case of the system with $G=-\theta,~a=\theta\mu$ where $\theta$ and $\mu$ are constants. The system turns to
\begin{equation}\label{exp-OU}
\begin{split}
dX_t=\theta(\mu-X_t)dt+\sigma dW_t,\quad X_0=x_0\in \mathbb{R}.
\end{split}
\end{equation}
The solution of \eqref{exp-OU} is called an Ornstein--Uhlenbeck (OU) process \cite{Oks03}. OU processes are usually used in finance to model the spread of stocks or to calculate interest rates and currency exchange rates. They also appear in physics to model the motion of a particle under friction.

On the one hand, the potential function of the system \eqref{exp-OU}, as in \eqref{potential} satisfies $U''(x) = -\theta$. So the OM functional of the system is
\begin{equation*}
S^{OM}(\psi)=\frac{1}{2\sigma^2}\int_{0}^T \left[ (\dot{\psi}-\theta\mu+\theta\psi)^2 - \theta \right] dt.
\end{equation*}
The corresponding Euler--Lagrange (EL) equation reads
\begin{equation} \label{OUEL}
\begin{cases}
\ddot{\psi}+\theta^2(\mu-\psi)=0,\\
\psi(0)=x_0,\quad\psi(T)=x_T.
\end{cases}
\end{equation}
Hence, we can solve the second-order boundary value problem \eqref{OUEL}, if it is uniquely solvable, to obtain the most probable transition path.

On the other hand, the transition probability density function is \cite{Duan2015}
\begin{equation*}
\begin{split}
p(x_T,T|x,t)=\frac{\sqrt{\theta}}{\sigma\sqrt{\pi(1-e^{-2\theta(T-t)})}}\exp\left\{-\frac{\theta}{\sigma^2}\frac{\left[x_T-(e^{-\theta (T-t)}x+\mu-\mu e^{-\theta (T-t)})\right]^2}{1-e^{-2\theta (T-t)}}\right\},
\end{split}
\end{equation*}
which leads to a Markovian bridge process as in \eqref{newsde}.
Then by Corollary \ref{linearMPTP}, the most probable transition path of \eqref{exp-OU} is described by the following first-order ODE,
\begin{equation}\label{OU}
\begin{split}
\frac{d\psi^*}{dt}=\theta(\mu-\psi^*)+2\theta e^{-\theta (T-t)}\frac{x_T-\left(e^{-\theta (T-t)}\psi^*+\mu-\mu e^{-\theta (T-t)}\right)}{1-e^{-2\theta (T-t)}},\quad \psi^*(0)=x_0.
\end{split}
\end{equation}
As \eqref{OU} is a one-dimensional first-order linear ODE, its analytical solution can be explicitly found as follows,
\begin{equation*}
\begin{split}
\psi^*(t)=\exp\left(-\int_0^tP(s)ds\right)\left[ x_0 + \int_0^tQ(s) \exp\left(\int_0^sP(u)du\right)ds \right],
\end{split}
\end{equation*}
where
\begin{equation*}
\begin{split}
P(t)=\ &-\theta-\frac{\theta e^{-2\theta (T-t)}}{1-e^{-2\theta (T-t)}},\\
Q(t)=\ & \theta\mu +\theta e^{-\theta (t-t)}\frac{x_T-(\mu-\mu e^{-\theta (T-t)})}{1-e^{-2\theta (T-t)}}.
\end{split}
\end{equation*}

\bigskip
Case 2. Consider the system \eqref{linearequation} in $\mathbb R^2$ with $a=(0,0)^T,~\sigma=1$ and
$$ G = \left(
\begin{array}{cc}
0 & 1 \\
1 & 0 \\
\end{array} \right).$$
Now the system turns to the following coupled system,
\begin{equation}\label{2dOU}\left\{
  \begin{aligned}
    dX^1_t &=X_t^2dt+ dW^1_t, \\
    dX^2_t &=X_t^1dt+ dW^2_t, \\
    X_0 &=x_0 \in \mathbb R^2.
  \end{aligned}\right.
\end{equation}

On the one hand, the potential function is given by $U(x^1,x^2)= - x^1x^2 + \mathrm{constant}$, so $\triangle U \equiv 0$. The OM action functional of \eqref{2dOU} is
\begin{equation*}
S^{OM}(\psi) 
= \frac{1}{2}\int_0^T  |\dot{\psi}(s)-G\psi(s)|^2 ds,
\end{equation*}
and the corresponding Euler--Lagrange equation is
\begin{equation}\label{2dOUEL}
\begin{cases}
\ddot{\psi}^1=\psi^1,\\
\ddot{\psi}^2=\psi^2,\\
\psi(0)=x_0,~\psi(T)=x_T.
\end{cases}
\end{equation}
Note that the two-dimensional boundary value problem \eqref{2dOUEL} can be separated into two independent one-dimensional boundary value problems. Thus we can use the shooting method to solve them one-by-one.

On the other hand, observe that
\begin{equation*}
e^{Gt} = \frac{1}{2}\left(
\begin{array}{cc}
e^{t}+e^{-t} & e^{t}-e^{-t} \\
e^{t}-e^{-t} & e^{t}+e^{-t} \\
\end{array} \right),\qquad
\left(e^{Gt}\right)^{-1} = e^{-Gt} = \frac{1}{2}\left(
\begin{array}{cc}
e^{t} + e^{-t} & -e^{t} + e^{-t} \\
-e^{t} + e^{-t} & e^{t} + e^{-t} \\
\end{array} \right).
\end{equation*}
So the mean $\mu(t)$ of $X_t$ is
$$\mu(t)=e^{Gt}x_0= \frac{1}{2}\left(
\begin{array}{c}
(e^{t}+e^{-t})x_0^1 + (e^{t}-e^{-t})x_0^2 \\
(e^{t}-e^{-t})x_0^1 + (e^{t}+e^{-t})x_0^2 \\
\end{array} \right),$$
and the covariance matrix $\Sigma(t)$ is
\begin{equation*}
\begin{split}
\Sigma (t) 
=\ &e^{Gt} \int_{0}^t\left(e^{Gs}\right)^{-1}\left[\left(e^{Gs}\right)^{-1}\right]^Tds \left(e^{Gt}\right)^T 
= \frac{1}{4}\left(
\begin{array}{cc}
e^{2t}-e^{-2t} & e^{-2t}-2+e^{2t} \\
e^{-2t}-2+e^{2t} & e^{2t}-e^{-2t}  \\
\end{array} \right),\\
\end{split}
\end{equation*}
with matrix inverse (when $t>0$):
\begin{equation*}
\begin{split}
\Sigma^{-1}(t)=\ & \left(
\begin{array}{cc}
-\frac{e^{2t}-e^{-2t}}{(1-e^{2t})(1-e^{-2t})} & \frac{e^{-2t}-2+e^{2t}}{(1-e^{2t})(1-e^{-2t})} \\
\frac{e^{-2t}-2+e^{2t}}{(1-e^{2t})(1-e^{-2t})} & -\frac{e^{2t}-e^{-2t}}{(1-e^{2t})(1-e^{-2t})} \\
\end{array} \right).\\
\end{split}
\end{equation*}
According to Corollary \ref{linearMPTP}, we know that the most probable transition path $\psi^* = (\psi^{*1}, \psi^{*2})^T$ of system \eqref{2dOU} solves the following system of first-order ODEs:
\begin{equation}\label{2dOUMPTP}\left\{
  \begin{aligned}
    \left( \begin{array}{c} \dot\psi^{*1}(t) \\ \dot\psi^{*2}(t) \end{array} \right) &= \left( \begin{array}{c} \psi^{*2}(t) \\ \psi^{*1}(t) \end{array} \right) + \left(e^{G(T-t)}\right)^T \Sigma^{-1}(T-t) \left(x_T-e^{G(T-t)}\psi^*(t) \right),\quad t\in[0,T),\\
   \psi^*(0)&=x_0.
  \end{aligned}\right.
\end{equation}
Note that the first-order ODE system \eqref{2dOUMPTP} is still coupled while the Euler--Lagrange equation \eqref{2dOUEL} is already decoupled.

Figure \ref{2dOUfigure} shows the numerical solutions of the Euler--Lagrange equation \eqref{2dOUEL} and the first-order ODE system \eqref{2dOUMPTP} with 
different transition times, by using the shooting method and forward Euler scheme respectively. They fit each other quite well. This verifies the validity of our results for the 2-d system \eqref{2dOU}.

\begin{figure}[ht]
  \centering
  \includegraphics[width=0.7\textwidth]{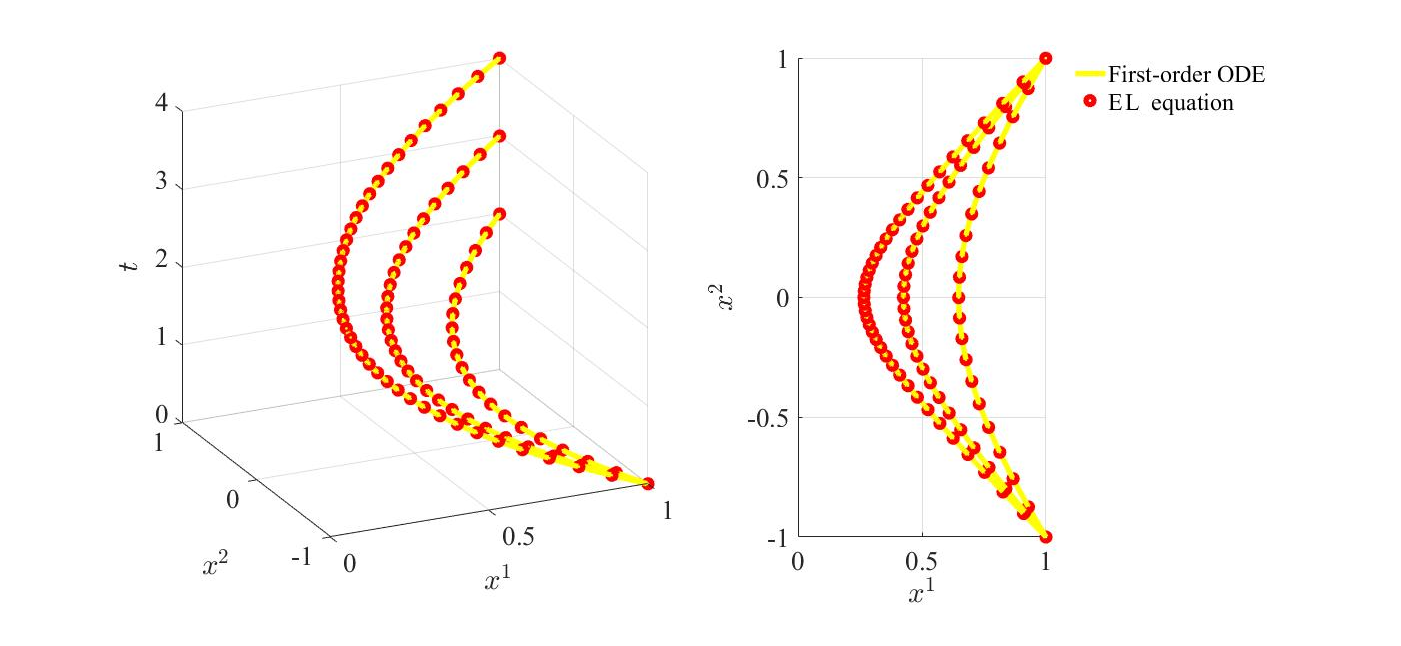}\\
  \caption{The most probable transition paths in $(x^1,x^2,t)$-plane and $(x^1,x^2)$-plane with initial and terminal conditions $x_0=(1,-1),~x_T=(1,1)$, under different transition times $T=2,3,4$. 
  The yellow lines are the numerical solutions of \eqref{2dOUMPTP} and the red point lines are the solutions of equation \eqref{2dOUEL}.}\label{2dOUfigure}
\end{figure}
\end{example}

\begin{example}[Hongler's model]
Consider Hongler's model \eqref{hongler} with parameters  $A=1$ and $B=-0.6$. The potential is
\begin{equation*} 
    \begin{aligned}
        U(x)=\frac{x^2}{2}-\ln \left[   {}_1F_1\left(\frac{1}{10};\frac{1}{2};x^2 \right) \right],
    \end{aligned}
\end{equation*}
and such $-U$ has the shape of double wells, which is shown in Figure \ref{HonglerU}. The system is
\begin{equation*} 
    \begin{aligned}
        dX_t=\left(-X_t +\frac{2}{5}\frac{ {}_1F_1\left(\frac{11}{10};\frac{3}{2};X_t^2 \right)}{{}_1F_1\left(\frac{1}{10};\frac{1}{2};X_t^2 \right)}X_t \right)dt+dW_t,\quad X_0=x_0.
    \end{aligned}
\end{equation*}
The two wells are located around $x=-1.4$ and $x=1.4$ respectively. We consider the transition of this model from $x_0=-1.4$ at time 0 to $x_T=1.4$ at time $T$. The MPTP satisfies the following ODE
\begin{equation*}
\begin{aligned}
    d\psi^*(t)=&\ \Bigg[ -\psi^*(t)+\frac{2e^{- (T-t)}\left(x_T-\psi^*(t)e^{- (T-t)}\right)}{1-e^{-2 (T-t)} }  \Bigg]dt,\quad t\in(0,T),\quad \psi^*(0)=-1.4.
\end{aligned}
\end{equation*}
Note that this ODE has the same form as the OU case \eqref{OU} with parameters $\theta=1$ and $\mu=0$. Thus, they share the same MPTPs (see our discussions in Subsection \ref{sec-4-2}).

\begin{figure}[htb]
  \centering
  \includegraphics[width=0.5\textwidth]{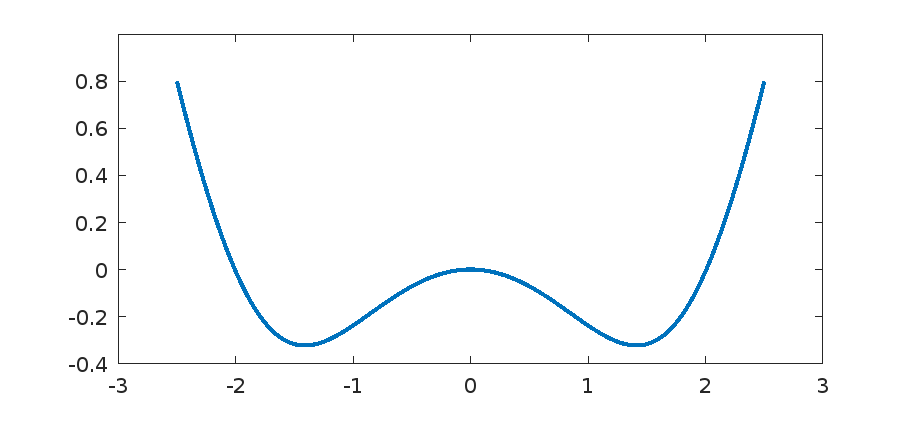}\\
  \caption{\small{The 
 shape of minus Hongler's potential $-U$ with $A=1$ and $B=-0.6$.}}\label{HonglerU}
\end{figure}
\end{example}

\begin{example}[Ginzburg--Landau double-well system]

Consider the following scalar double-well system:
\begin{equation*}
\begin{split}
dX_t=(X_t-X_t^3)dt+\sigma dW_t,\quad X_0=x_0\in \mathbb{R}.
\end{split}
\end{equation*}
From \cite[Theorems 2.2 \& 2.4]{xi2019jump} or \cite[Theorem 3.1]{albeverio2010existence}, we know that the above SDE has a pathwisely unique global strong solution.
It is easy to see that $1$ and $-1$ are stable equilibrium states of the associated deterministic system and thereby metastable states of the stochastic system, while $0$ is an unstable equilibrium state of the deterministic system. We consider the transition phenomena between metastable states $x_0=-1$ and $x_T=1$.

The first approximation $\psi_{\text{appr},1}$ of the most probable transition path is
\begin{equation}\label{psi1}
\begin{split}
\frac{d\psi_{\text{appr},1}}{dt}= \frac{x_T-\psi_{\text{appr},1}}{T-t}-\frac{1}{2}(T-t)\left[ (\psi_{\text{appr},1}-\psi_{\text{appr},1}^3)(1-3\psi_{\text{appr},1}^2) -3\sigma^2\psi_{\text{appr},1} \right],~t\in[0,T),\ \psi_{\text{appr},1}(0)=x_0.\\
\end{split}
\end{equation}
And the other one $\psi_{\text{appr},2}$ for small noise case is
\begin{equation}\label{psi2}
\begin{split}
\frac{d\psi_{\text{appr},2}}{dt}=\frac{x_T-\psi_{\text{appr},2}}{T-t}-(T-t)\int_{0}^{1}(1-u)(Z-Z^3)(1-3Z^2)du,~t\in[0,T),\ \psi_{\text{appr},2}(0)=x_0,
\end{split}
\end{equation}
where $Z=x_Tu+\psi_{\text{appr},2}(1-u)$.

The OM action functional of this system is
\begin{equation}\label{OM-GL}
\begin{split}
S^{OM}(\psi)=\frac{1}{2}\int_{0}^{T}\left[\frac{(\dot{\psi}-(\psi-\psi^3))^2}{\sigma^2}+(1-3\psi^2)\right]dt.
\end{split}
\end{equation}
The corresponding Euler--Lagrange equation reads
\begin{equation}\label{ELOM}
\begin{cases}
\ddot{\psi}=\left(\psi-\psi^3\right)(1-3\psi^2)-3\sigma^2\psi,\\
\psi(0)=x_0, \psi(T)=x_T,
\end{cases}
\end{equation}
A general numerical way to solve this second-order differential equation is the shooting method. We denote the path computed by the shooting method as $\psi_{\text{shoot}}$.

We compute the corresponding paths $\psi_{\text{appr},1}$, $\psi_{\text{appr},2}$ and $\psi_{\text{shoot}}$ respectively. Here we set the time step to be $\Delta t=10^{-3}$. The paths $\psi_{\text{appr},1}$ and $\psi_{\text{appr},2}$ can be numerically computed by forward Euler scheme according to equations $(\ref{psi1})$ and $(\ref{psi2})$. And we use the shooting method with Newton iteration to compute the path $\psi_{\text{shoot}}$, where we set the iteration error to be $10^{-3}$, i.e., we stop the iteration if $|\psi_{\text{shoot}}(T)-x_T|<10^{-3}$.

\begin{figure}[htb]
  \centering
  \includegraphics[width=0.5\textwidth]{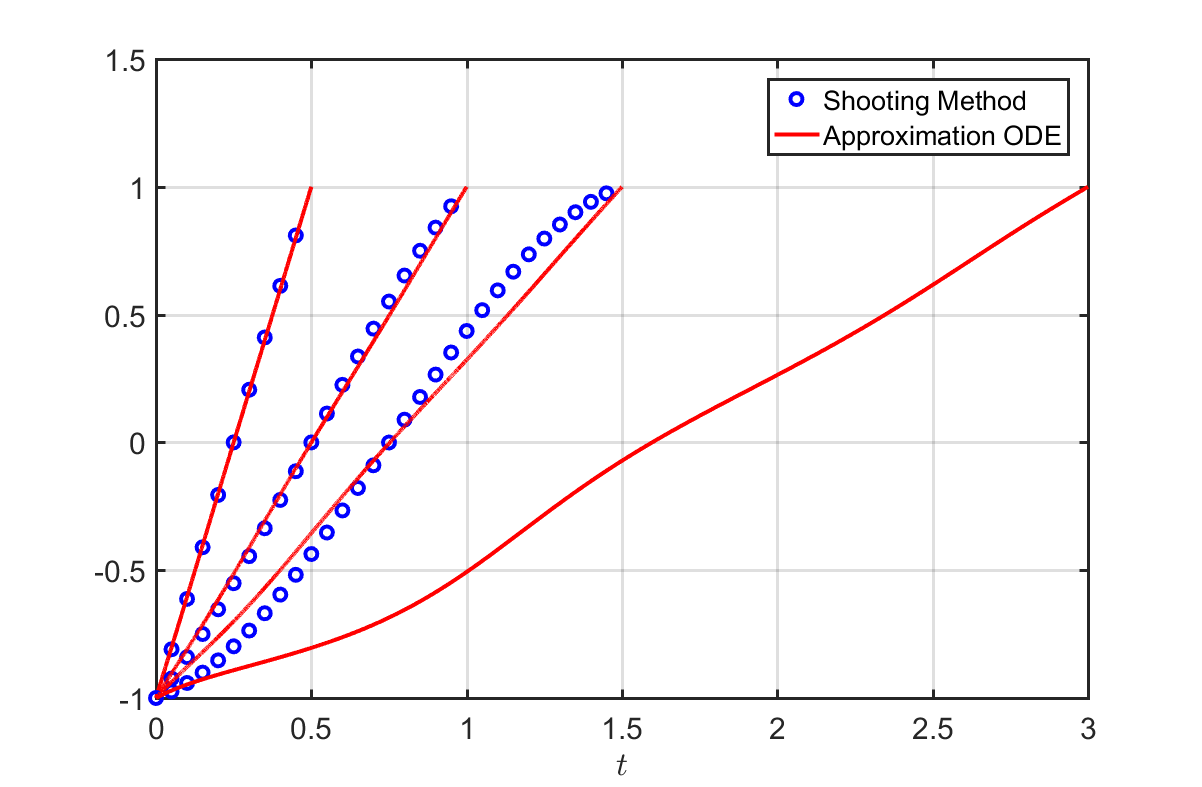}\\
  \caption{\small{The most probable transition paths approximated by approximation ODE \eqref{psi1} ($\psi_{\text{appr},1}$) and shooting method \eqref{ELOM} ($\psi_{\text{shoot}}$) with the same noise intensity $\sigma=1$, for different transition times $T=0.5,~1,~1.5,~3$.}}\label{doublewellOM}
\end{figure}

\begin{table}[htb]
\centering
\begin{tabular}{c|c|c|c|c}
  \hline
  $\mbox{Transition Time}$ & 0.5 & 1 & 1.5 & 3  \\
  \hline
  $S^{OM}(\psi_{\text{shoot}})$  & 4.0167 &  2.0108 & 1.2853 & NaN\\
   $S^{OM}(\psi_{\text{appr},1})$ &   4.0194 & 2.0345 & 1.3772 &  0.6288\\
  \hline
\end{tabular}
\caption{\small{The values of Onsager--Machlup functional for the most probable transition paths in Figure $\ref{doublewellOM}$.}}\label{MPTPOM}
\end{table}

Figure \ref{doublewellOM} shows the paths computed by \eqref{psi1} and \eqref{ELOM} under the noise intensity $\sigma=1$, for different transition times $T=0.5,~1,~1.5,~3$. We also compute the values of corresponding Onsager--Machlup functional of all these paths by discretization of \eqref{OM-GL}.

The plots in Figure \ref{doublewellOM} and their OM values in Table \ref{MPTPOM} show that when the transition time is short, the approximation ODE \eqref{psi1} is indeed applicable for simulating the MPTP of the system. 
It is worth noting that, in Figure \ref{doublewellOM} when $T=3$, the numerical shooting method failed to find the most probable transition path and thus the OM functional value in Table \ref{MPTPOM} is undefined and represented by ``NaN''. The shooting method did not work because the initial value of $\dot\psi$ we chose is not suitable and makes the iteration diverge. But the shooting method is not the only way to solve the boundary value problem. Other methods like the finite element method are also applicable with sufficient efficiency and accuracy, though this is beyond the topic of our paper.


Figure $\ref{doublewellFW}$ shows the paths computed by \eqref{psi2} and \eqref{ELOM} under different noise intensities $\sigma=1,~0.5,~0.1,$ and fixed transition time $T=2$. We also compute the values of the OM functional, listed in Table \ref{MPTPFW}. These numerical results show that when the noise intensity $\sigma$ is small,  \eqref{psi2} is suitable for the approximation of MPTPs, as expected in Subsection \ref{sec-4-3-2}.

\begin{figure}[htb]
  \centering
  \includegraphics[width=0.5\textwidth]{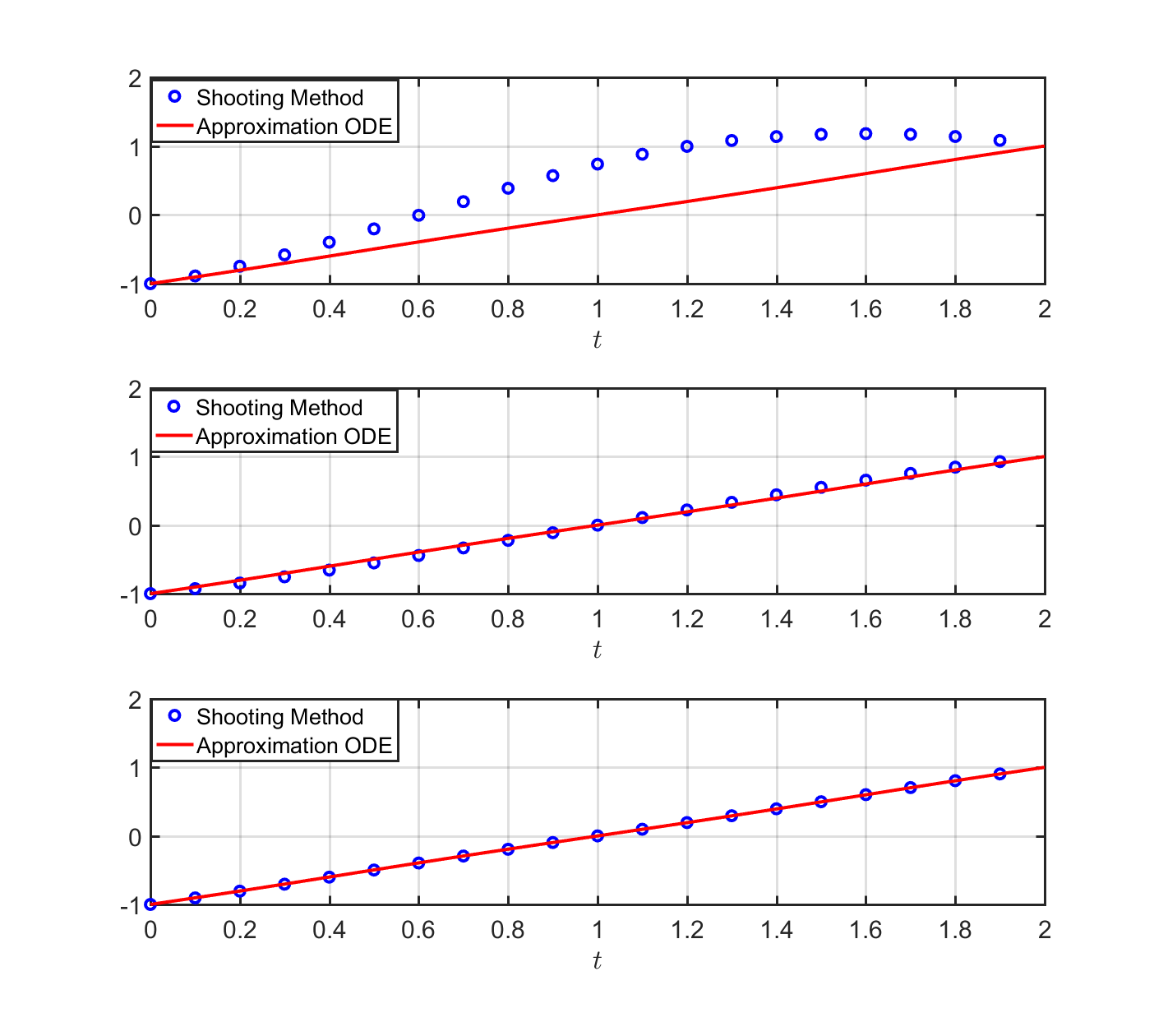}\\
  \caption{\small{The most probable transition paths approximated by approximation ODE \eqref{psi2} ($\psi_{\text{appr},2}$) and shooting method \eqref{ELOM} ($\psi_{\text{shoot}}$) under different noise intensities $\sigma=1$ (upper), 0.5 (middle),~0.1 (bottom), and fixed transition time $T=2$.}}\label{doublewellFW}
\end{figure}

\begin{table}[htb]
\centering
\begin{tabular}{c|c|c|c}
  \hline
  $\mbox{Noise Intensity}$ & 1 & 0.5 & 0.1 \\
  \hline
  $S^{OM}(\psi_{\text{shoot}})$  &  0.7059 &  4.2467 &  107.5956 \\
  $S^{OM}(\psi_{\text{appr},2})$ & 1.0761 & 4.3040 &   107.5974 \\
  \hline
\end{tabular}
\caption{\small{The values of Onsager--Machlup functional for the most probable transition paths in Figure $\ref{doublewellFW}$.}}\label{MPTPFW}
\end{table}

\end{example}

\section{Connections to other topics}\label{sampling&control}

In this section, we will briefly sketch the connections with other topics including trajectory sampling and optimal control, without going deep into them.

\subsection{Trajectory sampling}

As we have seen in Remark \ref{remark-1}-(i), the probabilistic configurations of strong solutions of two SDEs \eqref{SDE} and \eqref{newsde} are not necessarily the same. On the one hand, if we specify \emph{different} probability spaces for them, or if we just consider \emph{weak} solutions of them instead of strong ones, then the only relation between the solutions that we can talk about is their distributional information. Indeed, we have shown in Lemma \ref{equivalence} that the solution of the bridge SDE \eqref{newsde} shares the same law with the $(x_0,T,x_T)$-bridge derived from the original SDE \eqref{SDE}. However, in general, no trajectory information for their solutions can be clarified, or everything for their trajectories can be possible in this case.

On the other hand, if we choose the \emph{same} configuration for \eqref{SDE} and \eqref{newsde}, that is, $(\tilde{\Omega},\mathcal{\tilde F}, \{\tilde{\mathcal{F}}_t\}_{t\ge0}, \mathbb{\tilde P}) = (\Omega,\mathcal{F}, \{\mathcal{F}_t\}_{t\ge0}, \mathbb{P})$ and $\tilde W = W$, then we can construct the strong solution of \eqref{newsde} directly from that of \eqref{SDE}. More precisely, by replacing the noise terms in \eqref{SDE} and \eqref{newsde}, we obtain the following

\begin{corollary}[Trajectory relation between SDEs and their bridges]
Let $X$ and $Y$ be the unique strong solutions of SDEs \eqref{SDE} and \eqref{newsde}, respectively, on the same probability space $(\Omega,\mathcal{F}, \{\mathcal{F}_t\}_{t\ge0}, \mathbb{P}^{x_0})$ endowed with the same $k$-dimensional Brownian motion $W$. Then the following SDE holds,
\begin{equation}\label{X-Y}
   dY_t=\left[ \nabla U(X_t) - \nabla U(Y_t)+\sigma^2\nabla\ln p(x_T,T|Y_t,t) \right]dt + dX_t, \quad t\in[0,T).
\end{equation}
\end{corollary}

Even for a trajectory $X (\omega)$ that reaches $x_T$ at time $t=T$, the resulting trajectory $Y (\omega)$ realized from \eqref{X-Y} is \emph{not} the same as $X (\omega)$, since the extra drift $\sigma^2\nabla\ln p(x_T,T|Y_t(\omega),t)$ is not identically zero (see Example \ref{brownian} for an example of Brownian case).

The relation \eqref{X-Y} suggests a general approach for sampling trajectories of Markovian bridges: given a sample trajectory $X (\omega)$ of the strong solution of the original SDE \eqref{SDE}, one can use \eqref{X-Y} to sample a trajectory $Y (\omega)$ for the strong solution of bridge SDE \eqref{newsde} by many discrete-time schemes, e.g., Euler--Maruyama method, if the transition density $p$ is explicitly expressed.

In particular, for the Brownian and Ornstein--Uhlenbeck cases, many explicit representations for the solutions of bridge SDE \eqref{newsde} can be found in the literature \cite{barczy2013sample}. Consider the Ornstein--Uhlenbeck (OU) process satisfying the following scalar SDE (a special case of \eqref{exp-OU}):
\begin{equation}\label{OU-1}
dX_t= -\theta X_t dt+\sigma dW_t,\quad X_0=0\in \mathbb{R},
\end{equation}
with two constants $\sigma,\theta>0$. The standard Brownian motion can be regarded as the solution of \eqref{OU-1} with $\theta=0$ and $\sigma=1$.
Three representations, anticipative version, integral representation and space-time transform, are proposed in \cite{barczy2013sample}, respectively to construct Brownian and OU bridges. These representations are equal in law, as they are all solutions to the bridge SDE \eqref{newsde} for Brownian and OU cases respectively. Specifically, the integral representation is the unique strong solution to the associated bridge SDE, while the other two representations are only weak solutions. It was also demonstrated and partially proved that when a trajectory $X (\omega)$ of strong solution of \eqref{OU-1} reaches $0$ at time $t=T$, the sample trajectory realized from the integral representation as strong solution does not coincide with $X (\omega)$, while the sample trajectories realized from other two representations as weak solutions may and may not coincide. These match with our previous observations. 

Different from the bridge SDE construction and the three representations in \cite{barczy2013sample}, an intuitive and ingenious construction of Brownian and OU bridges was introduced in \cite[Appendix B]{pinski2021novel}, by taking a linear combination of two unconstrained independent trajectories and setting the endpoint of the linear combination to zero. That is, we suppose $X^1$ and $X^2$ to be two mutually independent copies of solution of \eqref{OU-1}, on probability spaces $(\Omega^1,\F^1,\{\F^1_t\},\mathbb{P}^1)$ and $(\Omega^2,\F^2,\{\F^2_t\},\mathbb{P}^2)$ respectively. Define a new process $X$ on the product probability space $(\Omega^1\times \Omega^2,\F^1 \otimes \F^2,\{\F^1_t \otimes \F^2_t\},\mathbb{P}^1\otimes \mathbb{P}^2)$ by
\begin{equation*}
  X(\omega_1, \omega_2) = X^1(\omega_1)\cos(\kappa(\omega_1, \omega_2)) + X^2(\omega_2)\sin(\kappa(\omega_1, \omega_2)),
\end{equation*}
where $\kappa$ is a random variable on the product space $\Omega^1\times \Omega^2$ determined by the following terminal condition,
\begin{align*}
X_T^1 \cos\kappa + X^2_T\sin\kappa =0,
\end{align*}
that is, $\kappa=\arctan(- X^1_T/X^2_T )$, with the conventions that $\arctan(0/0)=0$, $\arctan(+\infty)=\pi/2$ and $\arctan(-\infty)=-\pi/2$.
This construction implies that the new process $X$ satisfies $X_0 = X_T = 0$ and
\begin{equation}\label{X}
    dX_t = -\theta X_t dt+ \sigma d( W^1_t \cos\kappa + W_t^2\sin\kappa),
\end{equation}
We can prove that $X$ indeed has the same law with the $(0,T,0)$-bridge of the original OU process \eqref{OU-1}. See Appendix \ref{Bapp}, precisely Propositions \ref{BB} and \ref{OUB}. Since the noise $W^1 \cos\kappa + W^2 \sin\kappa$ in SDE \eqref{X} is not an $\{\F^1_t \otimes \F^2_t\}$-adapted process due to the fact that the random variable $\kappa$ depends on the information at time $t=T$, $X$ in \eqref{X} is only a weak solution of the bridge SDE associated to \eqref{OU-1}, which gives an alternative anticipative representation of $(0,T,0)$-OU bridges. 


\subsection{Optimal control perspective}

Due to Theorem \ref{maintheory}, the optimization problem
\begin{equation}\label{optimal-0}
\begin{split}
\inf_{\psi\in C^2_{x_0,x_T}[0,T]}S_{X}^{OM}(\psi)
\end{split}
\end{equation}
turns to another optimization problem
\begin{equation}\label{optimal}
  \inf_{\psi\in C^2_{x_0}[0,T]}S_{Y}^{OM}(\psi),
\end{equation}
where $S_{Y}^{OM}$ is the OM functional for the system \eqref{newsde} over time interval $[0,T]$, i.e.,
\begin{equation*}
\begin{split}
S_{Y}^{OM}(\psi) = \frac{1}{2}\int_0^{T}\left[\frac{|\dot{\psi}(s)-b(s,\psi(s))|^2}{\sigma^2}+\nabla\cdot b(s,\psi(s))\right]ds,
\end{split}
\end{equation*}
where $b$ is the modified drift in \eqref{m-drift}.
The Lagrangian functions of the actions $S_{X}^{OM}$ and $S_{Y}^{OM}$ are both of the form
\begin{equation}\label{Lagrangian-OM}
  L^{OM}(t,\psi,\dot \psi) = \frac{1}{2} |\dot \psi- v(t,\psi)|^2 + \frac{\sigma^2}{2} \nabla\cdot v(t,\psi)
\end{equation}
with $v = -\nabla U$ and $v=b$, respectively.

The two optimization problems \eqref{optimal-0} and \eqref{optimal} can be both translated to optimal control problems. To be precise, \eqref{optimal-0} can be translated into
\begin{equation*}
  \text{minimizing} \quad J_X[\psi,u] := \frac{1}{2}\int_0^T \left[ \frac{|u(s)|^2}{\sigma^2} - \triangle U(\psi(s)) \right] ds + \Theta(\psi(T)),
\end{equation*}
subject to
\begin{equation}\label{state}\left\{
\begin{aligned}
  \dot{\psi}(s) &= -\nabla U(\psi(s)) + u(s) ,\quad s\in[0,T], \\
  \psi(0) &= x_0,
\end{aligned}\right.
\end{equation}
where the functional $J_X$ is called the payoff functional, the function $\Theta$ is given by
\begin{equation*}
  \Theta(x) =
  \begin{cases}
    0, & x = x_T, \\
    +\infty, & \text{otherwise}.
  \end{cases}
\end{equation*}
and called the terminal cost, equation \eqref{state} is called the state equation.
Similarly, \eqref{optimal} is equivalent to
\begin{equation*}
  \text{minimizing} \quad J_Y[\psi,u] := \frac{1}{2}\int_0^T \left[ \frac{|u(s) - \sigma^2\nabla\ln p(x_T,T|\psi(s),s)|^2}{\sigma^2}+ \nabla\cdot b(s,\psi(s)) \right] ds,
\end{equation*}
subject to the same state equation \eqref{state}. The singularities of $\Theta$ and $p$ at $x_T$ imply that the minimizers of $J_X$ and $J_Y$ both satisfy $\psi(T) = x_T$. The equivalence between \eqref{optimal-0} and \eqref{optimal} provides an indirect way to obtain that $J_X$ and $J_Y$ share the same minimizer, which is not apparent from their forms. 

A recent work \cite{GLLL23} reinterpreted the transition-path theory of the original diffusion \eqref{SDE} from the energy basin of $x_0$ to that of $x_T$, as a stochastic control problem in infinite time horizon in terms of the Freidlin--Wentzel Lagrangian \eqref{Lagrangian-FW}, without the divergence term compared to the Onsager--Machlup Lagrangian \eqref{Lagrangian-OM} considered here. We also refer to \cite{selk2021information,HS12} for more discussions on optimal control problems arising from the trajectory or entropic information of processes. It is worth mentioning that the second author and his collaborator recently developed a framework of stochastic geometric mechanics \cite{HZ23} \red{revealing the underlying ``second-order'' geometric structures in stochastic optimal control, and applied it to Onsager's reciprocal theory in nonequilibrium thermodynamics \cite{HZ23a}.}


\section{Conclusion and discussion}\label{conclusions}

In most of the existing literature, the problem of finding the most probable transition paths of stochastic dynamical systems is solved by studying the Euler--Lagrange equation, a second-order ODE with two boundary values.  In this work, we showed that the most probable transition paths of a stochastic system coincide with the most probable paths of the corresponding Markovian bridge system, and they are determined by a first-order ODE. This provides a new insight for related topics. The result mainly relies on the derivation of Onsager--Machlup action functionals from bridge measures. This method is different from existing works of deriving Onsager--Machlup action functionals. The Markovian bridge system has an extra singular drift term that forces all sample trajectories to end at a given point. For linear systems and Hongler's model, the first-order ODEs were proved to imply the Euler--Lagrange equations. For general nonlinear stochastic systems, although it is not possible to get an explicit expression for this extra drift, there exist some approximations for it. 
We should notice that our first-order ODEs are sufficient and necessary descriptions of most probable transition paths, while Euler--Lagrange equations are only necessary but not sufficient. So in our framework, to find the most probable transition paths for a general nonlinear stochastic dynamical system, the only issue to be concerned about is how to approximate the extra drift and what the error is. There is no need to worry about the problem of local minimizers, since Theorem \ref{maintheory} promises that the paths are global minimizers. This is not the case for Euler--Lagrange equations. 

To summarize, in this paper, we first developed a new method to derive Onsager--Machlup action functionals for Markovian bridge measures whereas previous works cannot apply to such measures. Secondly, we show that for linear systems, nonlinear Hongler's model, and general nonlinear stochastic systems, the most probable transition paths can be determined or approximated by first-order ODEs with explicit forms. This provides a potential framework to study transitions of stochastic systems. 

\section*{Acknowledgements}
The authors would like to thank Dr.~Ying Chao, Dr.~Shuang Chen, Dr~Xiaoli Chen, Dr.~Meng Li, Dr.~Pingyuan Wei, Dr.~Wei Wei, Dr.~Ao Zhang and Prof.~J.-C. Zambrini for helpful discussions. We would like to thank the anonymous reviewers for their thoughtful comments and efforts towards improving our manuscript. The work of Y. Huang is partly supported by the NSFC grants 11531006 and 11771449. Y. Huang also would like to thank the support from his research group in the National University of Singapore during his postdoctoral period.
The work of Q. Huang is supported by FCT, Portugal, project PTDC/MAT-STA/28812/2017.

\section*{Data availability}

The data cannot be made publicly available upon publication because they are not available in a format that is sufficiently accessible or reusable by other researchers. The data that support the findings of this study are available upon reasonable request from the authors.

\footnotesize{
\bibliographystyle{unsrt}
\bibliography{ref}

\begin{thebibliography}{10}

\bibitem{OM53}
L.~Onsager and S.~Machlup.
\newblock Fluctuations and irreversible processes.
\newblock {\em Physical Review}, 91(6):1505--1512, 1953.

\bibitem{MO53}
S.~Machlup and L.~Onsager.
\newblock Fluctuations and irreversible processes. {II. Systems} with kinetic
  energy.
\newblock {\em Physical Review}, 91(6):1512--1515, 1953.

\bibitem{tisza1957fluctuations}
L.~Tisza and I.~Manning.
\newblock Fluctuations and irreversible thermodynamics.
\newblock {\em Physical Review}, 105(6):1695, 1957.

\bibitem{Durr1978}
D.~D{\"u}rr and A.~Bach.
\newblock {The Onsager-Machlup function as Lagrangian for the most probable
  path of a diffusion process}.
\newblock {\em Communications in Mathematical Physics}, 60(2):153--170, 1978.

\bibitem{Ikeda1980}
N.~Ikeda and S.~Watanabe.
\newblock {\em Stochastic differential equations and diffusion processes},
  volume~24.
\newblock North-Holland Publishing Company, 2nd edition, 1989.

\bibitem{Kath1981path}
K.L.C. Hunt and J.~Ross.
\newblock Path integral solutions of stochastic equations for nonlinear
  irreversible processes: The uniqueness of the thermodynamic {Lagrangian}.
\newblock {\em The Journal of Chemical Physics}, 75(2):976--984, 1981.

\bibitem{Schulman1981}
L.S. Schulman.
\newblock {\em Techniques and applications of path integration}.
\newblock Dover Publications, Inc., 2005.

\bibitem{Wiegel1986}
F.W. Wiegel.
\newblock {\em Introduction to path-integral methods in physics and polymer
  science}.
\newblock World Scientific Publishing, 1986.

\bibitem{Wio2013}
H.S. Wio.
\newblock {\em Path integrals for stochastic processes: An introduction}.
\newblock World Scientific Publishing, 2013.

\bibitem{Khandekar2000}
D.C. Khandekar, S.V. Lawande, and K.V. Bhagwat.
\newblock {\em Path integral methods and their applications}.
\newblock World Scientific Publishing, 1993.

\bibitem{KA20}
J.~Kappler and R.~Adhikari.
\newblock Stochastic action for tubes: Connecting path probabilities to
  measurement.
\newblock {\em Physical Review Research}, 2(2):023407, 2020.

\bibitem{Lu2017}
Y.~Lu, A.~Stuart, and H.~Weber.
\newblock Gaussian approximations for transition paths in {Brownian} dynamics.
\newblock {\em SIAM Journal on Mathematical Analysis}, 49(4):3005--3047, 2017.

\bibitem{EVE10}
W.~E and E.~Vanden-Eijnden.
\newblock Transition-path theory and path-finding algorithms for the study of
  rare events.
\newblock {\em Annual review of physical chemistry}, 61:391--420, 2010.

\bibitem{GLLL23}
Y.~Gao, T.~Li, X.~Li, and J.-G. Liu.
\newblock Transition path theory for {Langevin} dynamics on manifolds: Optimal
  control and data-driven solver.
\newblock {\em Multiscale Modeling \& Simulation}, 21(1):1--33, 2023.

\bibitem{Zuckerman2000}
D.M. Zuckerman and T.B. Woolf.
\newblock Efficient dynamic importance sampling of rare events in one
  dimension.
\newblock {\em Physical Review E}, 63(1):016702, 2000.

\bibitem{Gobbo2012}
G.~Gobbo, A.~Laio, A.~Maleki, and S.~Baroni.
\newblock Absolute transition rates for rare events from dynamical decoupling
  of reaction variables.
\newblock {\em Physical Review Letters}, 109(15):150601, 2012.

\bibitem{Chao2019}
Y.~Chao and J.~Duan.
\newblock The {Onsager-Machlup} function as {Lagrangian} for the most probable
  path of a jump-diffusion process.
\newblock {\em Nonlinearity}, 32(10):3715, 2019.

\bibitem{Huang2019}
Y.~Huang, Y.~Chao, S.~Yuan, and J.~Duan.
\newblock Characterization of the most probable transition paths of stochastic
  dynamical systems with stable {L\'e}vy noise.
\newblock {\em Journal of Statistical Mechanics: Theory and Experiment},
  2019(6):063204, 2019.

\bibitem{Faccioli2006}
P.~Faccioli, M.~Sega, F.~Pederiva, and H.~Orland.
\newblock Dominant pathways in protein folding.
\newblock {\em Physical Review Letters}, 97(10):108101, 2006.

\bibitem{Wang2006}
J.~Wang, K.~Zhang, H.~Lu, and E.~Wang.
\newblock Dominant kinetic paths on biomolecular binding-folding energy
  landscape.
\newblock {\em Physical Review Letters}, 96(16):168101, 2006.

\bibitem{selk2021information}
Z.~Selk, W.~Haskell, and H.~Honnappa.
\newblock {Information projection on Banach spaces with applications to state
  independent KL-weighted optimal control}.
\newblock {\em Applied Mathematics \& Optimization}, 84(1):805--835, 2021.

\bibitem{bierkens2014explicit}
J.~Bierkens and H.J. Kappen.
\newblock Explicit solution of relative entropy weighted control.
\newblock {\em Systems \& Control Letters}, 72:36--43, 2014.

\bibitem{dashti2013map}
M.~Dashti, K.J.H. Law, A.M. Stuart, and J.~Voss.
\newblock {MAP estimators and their consistency in Bayesian nonparametric
  inverse problems}.
\newblock {\em Inverse Problems}, 29(9):095017, 2013.

\bibitem{AKLS21}
B.~Ayanbayev, I.~Klebanov, H.C. Lie, and T.J. Sullivan.
\newblock {$\Gamma$-convergence of Onsager--Machlup functionals: I. With
  applications to maximum a posteriori estimation in Bayesian inverse
  problems}.
\newblock {\em Inverse Problems}, 38, 2021.

\bibitem{Zeitouni1989}
O.~Zeitouni.
\newblock {On the Onsager-Machlup functional of diffusion processes around non
  $C^2$ curves}.
\newblock {\em The Annals of Probability}, 17(3):1037--1054, 1989.

\bibitem{Zeitouni1987}
O.~Zeitouni and A.~Dembo.
\newblock A maximum a posteriori estimator for trajectories of diffusion
  processes.
\newblock {\em Stochastics: An International Journal of Probability and
  Stochastic Processes}, 20(3):221--246, 1987.

\bibitem{Zeitouni1988}
O.~Zeitouni and A.~Dembo.
\newblock An existence theorem and some properties of maximum a posteriori
  estimators of trajectories of diffusions.
\newblock {\em Stochastics: An International Journal of Probability and
  Stochastic Processes}, 23(2):197--218, 1988.

\bibitem{DLL21}
Q.~Du, T.~Li, X.~Li, and W.~Ren.
\newblock The graph limit of the minimizer of the {Onsager-Machlup} functional
  and its computation.
\newblock {\em Science China Mathematics}, 64(2):239--280, 2021.

\bibitem{Li2021}
T.~Li and X.~Li.
\newblock Gamma-limit of the {Onsager--Machlup} functional on the space of
  curves.
\newblock {\em SIAM Journal on Mathematical Analysis}, 53(1):1--31, 2021.

\bibitem{E2002}
W.~E, W.~Ren, and E.~Vanden-Eijnden.
\newblock String method for the study of rare events.
\newblock {\em Physical Review B}, 66(5):052301, 2002.

\bibitem{E2004}
W.~E, W.~Ren, and E.~Vanden-Eijnden.
\newblock Minimum action method for the study of rare events.
\newblock {\em Communications on pure and applied mathematics}, 57(5):637--656,
  2004.

\bibitem{Soskin2006}
S.M. Soskin.
\newblock Most probable transition path in an overdamped system for a finite
  transition time.
\newblock {\em Physics Letters A}, 353(4):281--290, 2006.

\bibitem{Karatzas1991}
I.~Karatzas and S.~Shreve.
\newblock {\em Brownian motion and stochastic calculus}, volume 113.
\newblock Springer-Verlag New York, Inc., 1991.

\bibitem{Ethier1986}
S.N. Ethier and T.G. Kurtz.
\newblock {\em Markov processes: characterization and convergence}.
\newblock John Wiley \& Sons, Inc., 2nd edition, 2009.

\bibitem{Bogachev2015}
V.I. Bogachev, N.V. Krylov, M.~R{\"o}ckner, and S.V. Shaposhnikov.
\newblock {\em Fokker--Planck--Kolmogorov equations}, volume 207.
\newblock American Mathematical Society, 2015.

\bibitem{Dekker78}
H.~Dekker.
\newblock Diffusion processes and their paths.
\newblock {\em Physics Letters A}, 80(2):99--101, 1978.

\bibitem{MP16}
P.J. Malsom and F.J. Pinski.
\newblock Role of {Ito's} lemma in sampling pinned diffusion paths in the
  continuous-time limit.
\newblock {\em Physical Review E}, 94(4):042131, 2016.

\bibitem{Chaumont2011}
L.~Chaumont and G.U. Bravo.
\newblock Markovian bridges: weak continuity and pathwise constructions.
\newblock {\em The Annals of Probability}, 39(2):609--647, 2011.

\bibitem{Fitzsimmons1993}
P.~Fitzsimmons, J.~Pitman, and M.~Yor.
\newblock Markovian bridges: construction, palm interpretation, and splicing.
\newblock In {\em Seminar on Stochastic Processes, 1992}, pages 101--134.
  Springer, 1993.

\bibitem{Cetin2016}
U.~{\c{C}}etin and A.~Danilova.
\newblock Markov bridges: {SDE} representation.
\newblock {\em Stochastic Processes and their Applications}, 126(3):651--679,
  2016.

\bibitem{Doob1957}
J.L. Doob.
\newblock Conditional brownian motion and the boundary limits of harmonic
  functions.
\newblock {\em Bulletin de la Soci{\'e}t{\'e} Math{\'e}matique de France},
  85:431--458, 1957.

\bibitem{Metafune2011}
G.~Metafune, E.M. Ouhabaz, and D.~Pallara.
\newblock Long time behavior of heat kernels of operators with unbounded drift
  terms.
\newblock {\em Journal of mathematical analysis and applications},
  377(1):170--179, 2011.

\bibitem{Pinski2012}
F.J. Pinski, A.M. Stuart, and F.~Theil.
\newblock {$\Gamma$}-limit for transition paths of maximal probability.
\newblock {\em Journal of Statistical Physics}, 146(5):955--974, 2012.

\bibitem{Hairer2007}
M.~Hairer, A.M. Stuart, and J.~Voss.
\newblock Analysis of {SPDEs} arising in path sampling part {II}: The nonlinear
  case.
\newblock {\em The Annals of Applied Probability}, 17(5-6):1657--1706, 2007.

\bibitem{PS10}
F.J. Pinski and A.M. Stuart.
\newblock Transition paths in molecules at finite temperature.
\newblock {\em The Journal of Chemical Physics}, 132:184104, 2010.

\bibitem{HCW21}
Y.~Huang, Y.~Chao, W.~Wei, and J.~Duan.
\newblock Estimating the most probable transition time for stochastic dynamical
  systems.
\newblock {\em Nonlinearity}, 34(7):4543, 2021.

\bibitem{HV08}
M.~Heymann and E.~Vanden-Eijnden.
\newblock The geometric minimum action method: A least action principle on the
  space of curves.
\newblock {\em Communications on Pure and Applied Mathematics}, 61:1052--1117,
  2008.

\bibitem{BS20}
P.E. Bloomfield and W.A. Steele.
\newblock Thermodynamic processes.
\newblock {\em AccessScience}, 2023.

\bibitem{hongler1981study}
M.-O. Hongler.
\newblock Study of a class of non-linear stochastic processes boomerang
  behaviour of the mean path.
\newblock {\em Physica D: Nonlinear Phenomena}, 2(2):353--369, 1981.

\bibitem{garrido1985exact}
L.~Garrido and J.~Masoliver.
\newblock Exact temporal evolution for some nonlinear diffusion processes.
\newblock {\em Journal of mathematical physics}, 26(3):522--527, 1985.

\bibitem{hongler1982exact}
M.-O. Hongler and W.M. Zheng.
\newblock Exact solution for the diffusion in bistable potentials.
\newblock {\em Journal of Statistical Physics}, 29(2):317--327, 1982.

\bibitem{romero1990diffusion}
J.L. Romero, J.~Ramirez, F.~Romero, and J.~Archilla.
\newblock Diffusion in a class of exactly solvable non-harmonic potentials.
  intrinsic effects induced by non-linearities.
\newblock {\em Physica D: Nonlinear Phenomena}, 41(1):79--88, 1990.

\bibitem{Orland2011}
H.~Orland.
\newblock Generating transition paths by {Langevin} bridges.
\newblock {\em The Journal of Chemical Physics}, 134(17):174114, 2011.

\bibitem{ait1999transition}
Y.~A{\"\i}t-Sahalia.
\newblock Transition densities for interest rate and other nonlinear
  diffusions.
\newblock {\em The journal of finance}, 54(4):1361--1395, 1999.

\bibitem{choi2015explicit}
S.~Choi.
\newblock Explicit form of approximate transition probability density functions
  of diffusion processes.
\newblock {\em Journal of Econometrics}, 187(1):57--73, 2015.

\bibitem{Freidlin2012}
M.I. Freidlin and A.D. Wentzell.
\newblock {\em Random perturbations of dynamical systems}.
\newblock Springer-Verlag Berlin Heidelberg, 3rd edition, 2012.

\bibitem{Delarue2017}
M.~Delarue, P.~Koehl, and H.~Orland.
\newblock Ab initio sampling of transition paths by conditioned {Langevin}
  dynamics.
\newblock {\em The Journal of Chemical Physics}, 147(15):152703, 2017.

\bibitem{Oks03}
B.~$\emptyset$ksendal.
\newblock {\em Stochastic differential equations: an introduction with
  applications}.
\newblock Springer-Verlag Berlin Heidelberg, 6th edition, 2003.

\bibitem{Duan2015}
J.~Duan.
\newblock {\em An introduction to stochastic dynamics}, volume~51.
\newblock Cambridge University Press, 2015.

\bibitem{xi2019jump}
F.~Xi and C.~Zhu.
\newblock Jump type stochastic differential equations with non-lipschitz
  coefficients: non-confluence, feller and strong feller properties, and
  exponential ergodicity.
\newblock {\em Journal of Differential Equations}, 266(8):4668--4711, 2019.

\bibitem{albeverio2010existence}
S.~Albeverio, Z.~Brze{\'z}niak, and J.~Wu.
\newblock Existence of global solutions and invariant measures for stochastic
  differential equations driven by poisson type noise with non-lipschitz
  coefficients.
\newblock {\em Journal of Mathematical Analysis and Applications},
  371(1):309--322, 2010.

\bibitem{barczy2013sample}
M.~Barczy and P.~Kern.
\newblock Sample path deviations of the {Wiener and the Ornstein--Uhlenbeck}
  process from its bridges.
\newblock {\em Brazilian Journal of Probability and Statistics},
  27(4):437--466, 2013.

\bibitem{pinski2021novel}
F.J. Pinski.
\newblock A novel hybrid {Monte Carlo} algorithm for sampling path space.
\newblock {\em Entropy}, 23(5):499, 2021.

\bibitem{HS12}
C.~Hartmann and C.~Sch{\"u}tte.
\newblock Efficient rare event simulation by optimal nonequilibrium forcing.
\newblock {\em Journal of Statistical Mechanics: Theory and Experiment},
  2012(11):P11004, 2012.

\bibitem{HZ23}
Q.~Huang and J.-C. Zambrini.
\newblock From second-order differential geometry to stochastic geometric
  mechanics.
\newblock {\em Journal of Nonlinear Science}, 33(4):67, 2023.

\bibitem{HZ23a}
Q.~Huang and J.-C. Zambrini.
\newblock Stochastic geometric mechanics in nonequilibrium thermodynamics:
  {Schr\"odinger} meets {Onsager}.
\newblock {\em Journal of Physics A: Mathematical and Theoretical},
  56(13):134003, 2023.

\bibitem{AS64}
M.~Abramowitz and I.A. Stegun.
\newblock {\em Handbook of mathematical functions with formulas, graphs, and
  mathematical tables}, volume~55.
\newblock US Government printing office, 1964.

\end{thebibliography}
}

\begin{appendices}
\section{New constructions for one-dimensional \texorpdfstring{$(0,T,0)$}{Lg}-Browinan and -OU bridges}\label{Bapp}
In \cite{barczy2013sample}, three representations (anticipative version, integral representation, and space-time transform) are proposed to sample the bridge processes of Brownian motions and OU processes. These representations are equal in law and are solutions of the associated bridge SDE. Specifically, the integral representation is the unique strong solution, while the anticipative version and space-time transform representations are weak solutions. 

A new way (different from the above three representations) of sampling one-dimensional Brownian or OU bridges was proposed in \cite{pinski2021novel}. The OU process considered satisfies the following one-dimensional SDE:
\begin{equation}\label{OUapp}
\begin{split}
dZ_t=-\theta Z_tdt + \sigma dW_t,\quad Z_0=0,
\end{split}
\end{equation}
where $W$ is a standard Brownian motion.
By admitting $\theta=0$, \eqref{OUapp} can also cover Brownian motions with diffusivity $\sigma^2$. So in the sequel, we will regard Brownian motions as OU processes with vanishing $\theta$.
Suppose that $Z^1$ and $Z^2$ are two mutually independent OU processes, solving \eqref{OUapp} governed by independent standard Brownian motions $W^1$ and $W^2$ respectively. Let $\kappa$ be a real-valued random variable determined almost surely by
\begin{align*}
Z_T^1 \cos\kappa + Z^2_T \sin\kappa =0.
\end{align*}
Construct a new process 
\begin{equation}\label{OU-4}
  Z = Z^1 \cos\kappa + Z^2 \sin\kappa.
\end{equation}
Then clearly, $Z$ satisfies $Z_0 = Z_T = 0$ and, when $Z^1, Z^2$ are OU processes, the following SDE:
\begin{equation*}
    dZ_t =-\theta Z_t dt + \sigma d( W^1_t \cos\kappa + W_t^2 \sin\kappa).
\end{equation*}
Note that the noise $W^1 \cos\kappa + W^2 \sin\kappa$ is not a Brownian motion since $\kappa$ contains information from $\mathcal F_T$, so that it is not apparent whether $Z$ is a $(0,T,0)$-OU bridge or not.

The author of \cite{pinski2021novel} claimed that the process $Z$ constructed above is a $(0,T,0)$-OU bridge, without proof.
In this section, we shall prove that such construction is indeed an OU bridge.


\subsection{Proof of the Brownian case}
Put $\theta=0$ and $\sigma=1$ in \eqref{OUapp}. The above construction reads
\begin{equation}\label{constructedBB}
Z_t=\frac{W^2_T}{\sqrt{(W^1_T)^2+(W^2_T)^2}} W^1_t- \frac{W^1_T}{\sqrt{(W^1_T)^2+(W^2_T)^2}} W_t^2.
\end{equation}
Let $p_W(y,t|x,s)$ be the probability transition density of the standard Brownian motion from $x$ at time $s$ to $y$ at time $t$. 
Set $x_0 = y_0 = 0$. For $0 = t_0<t_1<t_2<\cdots<t_n<t_{n+1} = T$ where $n\in\mathbb{N}_+$, the independence of $W^1$ and $W^2$ implies that the joint probability density function of $(W^1_{t_1},\cdots, W^1_{t_n}, W^1_{T}, W^2_{t_1},\cdots, W^2_{t_n}, W^2_{T})$ is given by
\begin{equation*}
f_{W^1_{t_1},\cdots,W^1_{t_n},W^1_T,W^2_{t_1},\cdots,W^2_{t_n},W^2_T}(x_1,\cdots,x_{n+1},y_1,\cdots,y_{n+1}) = \prod_{i=1}^{n+1} p_W(x_i,t_i|x_{i-1},t_{i-1}) p_W(y_i,t_i|y_{i-1},t_{i-1})
\end{equation*}

Make the following change of variables: 
\begin{equation}\label{change-variables}
  \begin{split}
    y_i&=y_i, \quad i=0,\cdots,n+1, \\
    x_i &= x_i(z_i,x_{n+1},y_i,y_{n+1})= \begin{cases} \frac{z_i\sqrt{x_{n+1}^2+y_{n+1}^2} + x_{n+1}y_i}{y_{n+1}}, & i=1,\cdots,n,\\
    x_i, & i=0,n+1,
    \end{cases}
  \end{split}
\end{equation}
It is easy to verify that the Jacobian determinant of the above change of variables is
$$
|J|=\left(\frac{\sqrt{x_{n+1}^2+y_{n+1}^2}}{|y_{n+1}|}\right)^n.$$
Then the joint probability density for $(Z_{t_1},~\cdots,Z_{t_n},~W^1_T,~W^2_{t_1},\cdots,W^2_{t_n},~W^2_T)$ is

\begin{equation}\label{joint-1}
  \begin{split}
&\ f_{Z_{t_1},\cdots,Z_{t_n},W^1_T,W^2_{t_1},\cdots,W^2_{t_n},W^2_T}(z_1,\cdots,z_n,x_{n+1},y_1,\cdots,y_{n+1}) \\
=&\  |J| \prod_{i=1}^{n+1} p_W\big(x_i(z_i,x_{n+1},y_i,y_{n+1}),t_i \big|x_{i-1}(z_{i-1},x_{n+1},y_{i-1},y_{n+1}),t_{i-1}\big) 
p_W(y_i,t_i |y_{i-1},t_{i-1}).
    \end{split}
\end{equation}
Before we go into the calculation details, we prove the following properties of the density function $p_W$. 

\begin{lemma}\label{densityproof}
Let $p_W$ be the transition density function of the standard Brownian motion. Then the following equations hold:
\begin{equation*}
  \begin{aligned}
    &\ \frac{\sqrt{x_{n+1}^2+y_{n+1}^2}}{|y_{n+1}|} p_W\big(x_i(z_i,x_{n+1},y_i,y_{n+1}),t_i \big|x_{i-1}(z_{i-1},x_{n+1},y_{i-1},y_{n+1}),t_{i-1}\big) p_W(y_i,t_i |y_{i-1},t_{i-1})\\
    = &\ \begin{cases}
    \tilde{p}_W \left(y_{i} + \frac{x_{n+1}z_i}{\sqrt{x_{n+1}^2+y_{n+1}^2}} ,t_i \bigg|y_{i-1} + \frac{x_{n+1}z_{i-1}}{\sqrt{x_{n+1}^2+y_{n+1}^2}} ,t_{i-1}\right)p(z_i,t_i|z_{i-1},t_{i-1}), &i=1,2,\cdots,n,\\
    \tilde{p}_W \left(y_{n+1},T\bigg|y_{n} + \frac{x_{n+1}z_n}{\sqrt{x_{n+1}^2+y_{n+1}^2}} ,t_n\right)p(0,T|z_n,t_n), & i=n+1.
    \end{cases}
  \end{aligned}
\end{equation*}
where $\tilde p_W$ is the transition density function of the Brownian motion whose variance at time $t$ is $\frac{y^2_{n+1}t}{x^2_{n+1}+y^2_{n+1}}$.
\end{lemma}

\begin{proof}
For the case $i=1$, we have \begin{equation*}
  \begin{aligned}
    &\ \frac{\sqrt{x_{n+1}^2+y_{n+1}^2}}{|y_{n+1}|} p_W\left(\frac{z_1\sqrt{x_{n+1}^2+y_{n+1}^2} + x_{n+1}y_1}{y_{n+1}},t_1\Bigg|0,0\right)p_W(y_1,t_1 |0,0)\\
    =&\ \frac{\sqrt{x_{n+1}^2+y_{n+1}^2}}{|y_{n+1}|}\frac{1}{2\pi t_1} \exp\left\{ -\frac{\left(  z_1\sqrt{x_{n+1}^2+y_{n+1}^2} + x_{n+1}y_1 \right)^2 + y_1^2y^2_{n+1}}{2y^2_{n+1}t_1} \right\}\\
    =&\ \sqrt{\frac{x_{n+1}^2+y_{n+1}^2}{2\pi t_1 y^2_{n+1}}} \exp\left\{ -\frac{x^2_{n+1}+y^2_{n+1} }{2y^2_{n+1}t_1} \left(y_1+\frac{x_{n+1}z_1}{\sqrt{x_{n+1}^2+y_{n+1}^2}}\right)^2   \right\} \frac{1}{\sqrt{2\pi t_1}} \exp\left( \frac{z_1^2}{2t_1} \right)\\ 
    =&\ \tilde{p}_W\left(y_1 +\frac{x_{n+1}z_1}{\sqrt{x_{n+1}^2+y_{n+1}^2}},t_1 \Bigg|0,0\right)p_W(z_1,t_1|0,0).
  \end{aligned}
\end{equation*}
When $2\le i\le n$, we use the fact that the density function $p_W$ is homogeneous in space and time and the result of $i=1$ to derive as follows,
\begin{equation*}
  \begin{aligned}
    &\ \frac{\sqrt{x_{n+1}^2+y_{n+1}^2}}{|y_{n+1}|} p_W\left(\frac{z_i\sqrt{x_{n+1}^2+y_{n+1}^2} + x_{n+1}y_i}{y_{n+1}},t_i\Bigg|\frac{z_{i-1}\sqrt{x_{n+1}^2+y_{n+1}^2} + x_{n+1}y_{i-1}}{y_{n+1}},t_{i-1}\right)p_W(y_i,t_i |y_{i-1},t_{i-1})\\
    =&\ \frac{\sqrt{x_{n+1}^2+y_{n+1}^2}}{|y_{n+1}|} p_W\left(\frac{(z_i-z_{i-1})\sqrt{x_{n+1}^2+y_{n+1}^2} + x_{n+1}(y_i-y_{i-1})}{y_{n+1}},t_i-t_{i-1}\Bigg|0,0\right)p_W(y_i-y_{i-1},t_i-t_{i-1} |0,0)\\
    =&\ \tilde{p}_W\left(y_{i}-y_{i-1} + \frac{x_{n+1}(z_i-z_{i-1})}{\sqrt{x_{n+1}^2+y_{n+1}^2}} ,t_i-t_{i-1}\Bigg|0,0\right)p_W(z_i-z_{i-1},t_i-t_{i-1}|0,0)\\
    =&\ \tilde{p}_W\left(y_{i} + \frac{x_{n+1}z_i}{\sqrt{x_{n+1}^2+y_{n+1}^2}} ,t_i\Bigg|y_{i-1} + \frac{x_{n+1}z_{i-1}}{\sqrt{x_{n+1}^2+y_{n+1}^2}} ,t_{i-1}\right)p_W(z_i,t_i|z_{i-1},t_{i-1}).
    \end{aligned}
\end{equation*}
The proof for the case $i=n+1$ is similar to the first one, as follows
\begin{equation*}
  \begin{aligned}
    &\ \frac{\sqrt{x_{n+1}^2+y_{n+1}^2}}{|y_{n+1}|} p_W\left(x_{n+1},T\Bigg|\frac{z_n\sqrt{x_{n+1}^2+y_{n+1}^2} + x_{n+1}y_n}{y_{n+1}},t_n\right) p_W(y_{n+1},T |y_n,t_n)\\
    =&\ \frac{\sqrt{x_{n+1}^2+y_{n+1}^2}}{|y_{n+1}|}\frac{1}{2\pi(T-t_n)} \exp \left\{ -\frac{1}{2(T-t_n)} \left(x_{n+1}- \frac{z_n\sqrt{x_{n+1}^2+y_{n+1}^2}}{y_{n+1}}-\frac{x_{n+1}y_n}{y_{n+1}}\right)^2 \right\} \exp \left( -\frac{(y_{n+1}-y_n)^2}{2(T-t_n)} \right)\\
    =&\ \sqrt{\frac{x_{n+1}^2+y_{n+1}^2}{2\pi (T-t_n) y_{n+1}^2}} \exp \left\{ -\frac{x_{n+1}^2+y_{n+1}^2}{2(T-t_n)y_{n+1}^2} \left(  y_{n+1}-y_n- \frac{x_{n+1}z_n}{\sqrt{x_{n+1}^2+y_{n+1}^2}}  \right)^2 \right\} \frac{1}{\sqrt{2\pi (T-t_n)}} \exp\left( \frac{z_n^2}{2(T-t_n)}\right) \\
    =&\ \tilde{p}_W\left(y_{n+1},T\Bigg|y_{n} + \frac{x_{n+1}z_n}{\sqrt{x_{n+1}^2+y_{n+1}^2}} ,t_n\right)p_W(0,T|z_n,t_n).
  \end{aligned}
\end{equation*}
The proof is completed.
\end{proof}

\begin{proposition}\label{BB}
The process \eqref{constructedBB} has the same law with the standard $(0,T,0)$-Brownian bridge. 
\end{proposition}

\begin{proof}
Due to \eqref{joint-1} and Lemma \ref{densityproof}, it is easy to obtain the finite-dimensional probability density for $Z$ by applying Chapman--Kolmogorov equations,
\begin{equation*}
  \begin{aligned}
    &\ f_{Z_{t_1},\cdots,Z_{t_n}}(z_1,\cdots,z_n) \\
    =&\ \int_{-\infty}^{\infty}\cdots\int_{-\infty}^{\infty} f_{Z_{t_1},\cdots,Z_{t_n},W^1_T,W^2_{t_1},\cdots,W^2_{t_n},W^2_T}(z_1,\cdots,z_n,x_{n+1},y_1,\cdots,y_{n+1})dx_{n+1}dy_{n+1}\cdots dy_1\\
    =&\ \int_{-\infty}^{\infty}\cdots\int_{-\infty}^{\infty}\left(\frac{\sqrt{x_{n+1}^2+y_{n+1}^2}}{|y_{n+1}|} \right)^n p_W\left(\frac{z_1\sqrt{x_{n+1}^2+y_{n+1}^2} + x_{n+1}y_1}{y_{n+1}},t_1\Bigg|0,0\right)\cdots\\
&\ \cdot p_W\left(x_{n+1},T\Bigg|\frac{z_n\sqrt{x_{n+1}^2+y_{n+1}^2} + x_{n+1}y_n}{y_{n+1}},t_n\right) p_W(y_1,t_1 |0,0)\cdots p_W(y_{n+1},T |y_n,t_n)dy_1\cdots dy_{n+1}dx_{n+1}\\
=&\ p_W(0,T|z_n,t_n)p_W(z_n,t_n|z_{n-1},t_{n-1})\cdots p_W(z_1,t_1|0,0)\int_{-\infty}^{\infty}\cdots\int_{-\infty}^{\infty}\left(\frac{\sqrt{x_{n+1}^2+y_{n+1}^2}}{|y_{n+1}|} \right)^{-1} \tilde{p}_W\left(y_1 +\frac{x_{n+1}z_1}{\sqrt{x_{n+1}^2+y_{n+1}^2}},t_1 \Bigg|0,0\right)\cdots\\
&\ \cdot \tilde{p}_W\left(y_{n+1},T\Bigg|y_{n} + \frac{x_{n+1}z_n}{\sqrt{x_{n+1}^2+y_{n+1}^2}} ,t_n\right) dy_1\cdots dy_{n+1}dx_1\cdots dx_{n+1}\\
=&\ p_W(0,T|z_n,t_n)p_W(z_n,t_n|z_{n-1},t_{n-1})\cdots p_W(z_1,t_1|0,0)\int_{-\infty}^{\infty}\int_{-\infty}^{\infty}\left(\frac{\sqrt{x_{n+1}^2+y_{n+1}^2}}{|y_{n+1}|} \right)^{-1} \tilde{p}_W\left(y_{n+1},T |0,0\right) dy_{n+1}dx_{n+1}\\
=&\ p_W(0,T|z_n,t_n)p_W(z_n,t_n|z_{n-1},t_{n-1})\cdots p_W(z_1,t_1|0,0)\int_{-\infty}^{\infty}\int_{-\infty}^{\infty}\frac{1}{\sqrt{2\pi T}}\exp\left(-\frac{x^2_{n+1}+y^2_{n+1}}{2T}\right)dy_{n+1}dx_{n+1}\\
=&\ \frac{p_W(0,T|z_n,t_n)p_W(z_n,t_n|z_{n-1},t_{n-1})\cdots p_W(z_1,t_1|0,0)}{p_W(0,T|0,0)}.
    \end{aligned}
\end{equation*}
This shows that the process \eqref{constructedBB} has the same finite-dimensional distribution with the standard $(0,T,0)$-Brownian bridge (cf. \eqref{finitedistribution}). The result follows from a standard argument of measure theory (cf. the proof of Lemma \ref{equivalence}).
\end{proof}

\subsection{Proof of the Ornstein--Uhlenbeck case}
Now we turn to the general case in which $\theta>0,~\sigma=1$. The  construction \eqref{OU-4} reads
\begin{equation}\label{constructedOU}
\begin{split}
Z_t=\frac{Z     ^2_T}{\sqrt{(Z^1_T)^2+(Z^2_T)^2}} Z^1_t- \frac{Z^1_T}{\sqrt{(Z^1_T)^2+(Z^2_T)^2}} Z_t^2,
\end{split}
\end{equation}
where $Z^1$ and $Z^2$ are two mutually independent solution processes of \eqref{OUapp}, i.e., Ornstein--Uhlenbeck processes. We denote the common transition probability density function of $Z^1$ and $Z^2$ by $p$, which is given by 
$$p(y,t|x,s)=\sqrt{\frac{\theta}{\pi(1-e^{-2\theta(t-s)})}}\exp\left( -  \frac{\theta(y-xe^{-\theta(t-s)})^2}{1-e^{-2\theta(t-s)}} \right).$$

%
Under the same change of variables as \eqref{change-variables}, we have
\begin{lemma}
    Let $p$ be the transition density function of the Ornstein--Uhlenbeck process \eqref{OUapp} with $\theta>0$ and $\sigma=1$. Then the following equations hold:
\begin{equation*}
  \begin{aligned}
    &\ \frac{\sqrt{x_{n+1}^2+y_{n+1}^2}}{|y_{n+1}|} p\big(x_i(z_i,x_{n+1},y_i,y_{n+1}),t_i \big|x_{i-1}(z_{i-1},x_{n+1},y_{i-1},y_{n+1}),t_{i-1}\big) p(y_i,t_i |y_{i-1},t_{i-1})\\
    = &\ \begin{cases}
    \tilde{p} \left(y_{i} + \frac{x_{n+1}z_i}{\sqrt{x_{n+1}^2+y_{n+1}^2}} ,t_i \bigg|y_{i-1} + \frac{x_{n+1}z_{i-1}}{\sqrt{x_{n+1}^2+y_{n+1}^2}} ,t_{i-1}\right)p(z_i,t_i|z_{i-1},t_{i-1}), &i=1,2,\cdots,n,\\
    \tilde{p} \left(y_{n+1},T\bigg|y_{n} + \frac{x_{n+1}z_n}{\sqrt{x_{n+1}^2+y_{n+1}^2}} ,t_n\right)p(0,T|z_n,t_n), & i=n+1.
    \end{cases}
  \end{aligned}
\end{equation*}
 where $\tilde p$ is the transition density function for the Ornstein--Uhlenbeck process solving \eqref{OUapp} with $\sigma^2= \frac{ y^2_{n+1}}{x^2_{n+1}+y^2_{n+1}}$ and the same $\theta$.
\end{lemma}
\begin{proof}
For the case $i=1$, we have
\begin{equation*}
  \begin{aligned}
    &\ \frac{\sqrt{x_{n+1}^2+y_{n+1}^2}}{|y_{n+1}|}p\left(\frac{z_1\sqrt{x_{n+1}^2+y_{n+1}^2} + x_{n+1}y_1}{y_{n+1}},t_1\Bigg|0,0\right)p(y_1,t_1 |0,0)\\
    =&\ \frac{\sqrt{ x_{n+1}^2+y_{n+1}^2} }{|y_{n+1}|}  \frac{\theta}{   \pi(1-e^{-2\theta t_1}) } \exp \left\{-\frac{\theta}{ 1-e^{- \theta t_1 } } \left(\frac{z_1\sqrt{x_{n+1}^2+y_{n+1}^2} + x_{n+1}y_1}{y_{n+1}}\right)^2 -\frac{\theta y_1^2}{ 1-e^{-2\theta t_1} }\right\}\\
    =&\ \sqrt{\frac{\theta(x_{n+1}^2+y_{n+1}^2)}{   \pi(1-e^{-2\theta t_1})y_{n+1}^2} }  \exp \left\{- \frac{\theta(x_{n+1}^2+y_{n+1}^2)}{(1-e^{-2\theta t_1})y_{n+1}^2} \left(y_1 + \frac{x_{n+1}z_1}{\sqrt{x_{n+1}^2+y_{n+1}^2}} \right)^2\right\} \sqrt{\frac{\theta}{ \pi(1-e^{-2\theta t_1})} } \exp\left( - \frac{\theta z_1^2}{ 1-e^{-2\theta t_1} } \right)\\
    =&\ \tilde{p}\left(y_1 +\frac{x_{n+1}z_1}{\sqrt{x_{n+1}^2+y_{n+1}^2}},t_1 \Bigg|0,0\right)p(z_1,t_1|0,0).
  \end{aligned}
\end{equation*}
For $2\leq i\leq n$, even though the density function $p$ is not homogeneous in space, different from the Brownian case, we can still transform to the form of the case $i=1$. That is,
\begin{equation*}
  \begin{aligned}
    &\ \frac{\sqrt{x_{n+1}^2+y_{n+1}^2}}{|y_{n+1}|}p\left(\frac{z_i\sqrt{x_{n+1}^2+y_{n+1}^2} + x_{n+1}y_i}{y_{n+1}},t_i\Bigg|\frac{z_{i-1}\sqrt{x_{n+1}^2+y_{n+1}^2} + x_{n+1}y_{i-1}}{y_{n+1}},t_{i-1}\right)p(y_i,t_i |y_{i-1},t_{i-1})\\
    =&\ \frac{\sqrt{x_{n+1}^2+y_{n+1}^2}}{|y_{n+1}|} \frac{\theta}{  \pi(1-e^{-2\theta (t_i-t_{i-1})}) } \\
    &\ \cdot \exp \left\{-\frac{\theta}{ 1-e^{-2\theta (t_i-t_{i-1}) }} \left[ \frac{z_i\sqrt{x_{n+1}^2+y_{n+1}^2} + x_{n+1}y_i}{y_{n+1}} - e^{-\theta(t_i-t_{i-1})} \left(\frac{z_{i-1}\sqrt{x_{n+1}^2+y_{n+1}^2} + x_{n+1}y_{i-1}}{y_{n+1}}\right) \right]^2 \right\}\\
         &\ \cdot \exp\left\{-\frac{\theta \left(y_i - e^{-\theta(t_i-t_{i-1})}y_{i-1}\right)^2}{ 1-e^{-2\theta (t_i-t_{i-1})} }\right\}\\
         =&\ \frac{\sqrt{x_{n+1}^2+y_{n+1}^2}}{|y_{n+1}|} p\left(\left( z_i-z_{i-1}e^{-\theta(t_i-t_{i-1})}\right)\frac{\sqrt{x_{n+1}^2+y_{n+1}^2}}{y_{n+1}}+\left(y_i-y_{i-1}e^{-\theta(t_i-t_{i-1})}\right)\frac{x_{n+1}}{y_{n+1}},t_i-t_{i+1}\Bigg|0,0\right)\\
         &\ \cdot p\left(y_i-y_{i-1}e^{-\theta(t_i-t_{i-1})},t_i-t_{i-1} \big|0,0\right) \\
         =&\ \tilde{p}\left(y_i-y_{i-1}e^{-\theta(t_i-t_{i-1})} + \left(z_i-z_{i-1}e^{-\theta(t_i-t_{i-1})}\right)\frac{x_{n+1}}{\sqrt{x_{n+1}^2+y_{n+1}^2}},t_i-t_{i-1} \Bigg|0,0\right)\\
         &\ \cdot p\left(z_i-z_{i-1}e^{-\theta(t_i-t_{i-1})},t_i-t_{i-1}\big|0,0 \right)\\
        =&\ \tilde{p}\left(y_{i} + \frac{x_{n+1}z_i}{\sqrt{x_{n+1}^2+y_{n+1}^2}} ,t_i \Bigg|y_{i-1} + \frac{x_{n+1}z_{i-1}}{\sqrt{x_{n+1}^2+y_{n+1}^2}} ,t_{i-1}\right)p(z_i,t_i|z_{i-1},t_{i-1}).
  \end{aligned}
\end{equation*}
The proof for the case $i=n+1$ is similar to the case $i=i$, and will be omitted here.
The proof is completed.
\end{proof}

As in the last subsection, a straightforward application of the above lemma yields the following proposition.
\begin{proposition}\label{OUB}
The process \eqref{constructedOU} has the same law with the corresponding $(0,T,0)$-Ornstein--Uhlenbeck bridge. 
\end{proposition}

\section{Derivations related to Hongler’s model}\label{appendixhongler}
The confluent hypergeometric function \cite{AS64} ${}_1F_1$ is given by
\begin{equation*}
    \begin{aligned}
        {}_1F_1(a;b;z) :=\sum_{n=0}^{\infty}\frac{a^{(n)}z^n}{b^{(n)}n!}, \quad a,b\in \R, z\in\mathbb C
    \end{aligned}
\end{equation*}
where for $n\in\N$ and a real number $a$, $a^{(n)}$ denotes the $n$-th rising factorial of $a$, defined by
\begin{equation*}
    \begin{aligned}
        a^{(0)}=&\ 1,\\
        a^{(n)}=&\ a(a+1)(a+2)\cdots (a+n-1).
    \end{aligned}
\end{equation*}
The confluent hypergeometric function satisfies the following relation
\begin{equation}\label{relation-hyperg}
    \begin{aligned}
        z\frac{d{}_1F_1}{dz}(a;b;z)=z\frac{a}{b}{}_1F_1(a+1;b+1,z).
    \end{aligned}
\end{equation}
Substituting \eqref{honglerU1st} and the density \eqref{densityhongler} to \eqref{MPTP}, and applying the relation \eqref{relation-hyperg}, we obtain that the MPTPs of Hongler's SDE \eqref{hongler} satisfy
\begin{equation*} 
\begin{aligned}
    d\psi^*(t)=&\ \Bigg[ -\sqrt{A}\psi^*(t) \\
    &\quad + \left( {}_1F_1\left(\frac{B}{4\sqrt{A}}+\frac{1}{4};\frac{1}{2};\sqrt{A}(\psi^*(t))^2 \right)\right)^{-1} \\
    &\ \qquad \times 2\sqrt{A} \psi^*(t) \left(\frac{B}{2\sqrt{A}}+\frac{1}{2}  \right)   {}_1F_1\left(\frac{B}{4\sqrt{A}}+\frac{5}{4};\frac{3}{2};\sqrt{A}(\psi^*(t))^2 \right) \Bigg)\\
    &\quad + \frac{2\sqrt{A}e^{-\sqrt{A}(T-t)}}{1-e^{-2\sqrt{A}(T-t)} } \left(x_T-\psi^*(t)e^{-\sqrt{A}(T-t)} \right) \\
&\quad - \left( \sqrt{ \frac{2\pi}{\sqrt{A}} \sinh\left(\sqrt{A}(T-t)\right) }    {}_1F_1\left(\frac{B}{4\sqrt{A}}+\frac{1}{4};\frac{1}{2};\sqrt{A}(\psi^*(t))^2 \right)\right)^{-1} \\
&\ \qquad \times 2\sqrt{A}\psi^*(t)\left(\frac{B}{2\sqrt{A}}+\frac{1}{2}  \right) \sqrt{ \frac{2\pi}{\sqrt{A}} \sinh\left(\sqrt{A}(T-t) \right) }    {}_1F_1\left(\frac{B}{4\sqrt{A}}+\frac{5}{4};\frac{3}{2};\sqrt{A}(\psi^*(t))^2 \right) \Bigg]dt \\
=&\ \left[ -\sqrt{A}\psi^*(t) + \frac{2\sqrt{A}e^{-\sqrt{A}(T-t)}}{1-e^{-2\sqrt{A}(T-t)} } \left(x_T-\psi^*(t)e^{-\sqrt{A}(T-t)} \right) \right]dt.
\end{aligned}
\end{equation*}

\end{appendices}

\end{document}